\documentclass[journal]{IEEEtran}

\usepackage{graphicx}
\usepackage{color}          % to use colors
\usepackage{amsfonts}       % to use eg mathbb
\usepackage{amsthm}         % to use theorem environments
\usepackage{mathtools}      % to use coloneqq
\usepackage{tikz}		    % tikz for graphs
\usepackage{pgfplots}       % to use axis environment in tikz
\usepackage{algorithm}      % to put algorithms
\usepackage{algpseudocode}
\usepackage{amsmath}
\usepackage{subcaption}			% to use subfigure 
\algnewcommand{\LeftComment}[1]{\Statex \(\triangleright\) #1} 
\graphicspath{{figures/}}
\DeclareRobustCommand{\rchi}{{\mathpalette\irchi\relax}}
\newcommand{\irchi}[2]{\raisebox{\depth}{$#1\chi$}}

\definecolor{Granata}{rgb}{0.64,0,0} 
\definecolor{QuasiBlue}{rgb}{0.03,0.3,0.72}

\newcommand{\mc}[1]{\mathcal{#1}}
\newcommand{\mb}[1]{\mathbb{#1}}
\newcommand{\N}{N}
\newcommand{\bk}{{{k_{\rm u}}}}
\newcommand{\poa}{{\rm PoA}}
\newcommand{\nashe}[1]{{\rm NE}({#1})}
\newcommand{\nashealg}[1]{\text{\rm NE}^{\rm alg}({#1})}
\newcommand{\gfk}{\mc{G}_f^k}

\newcommand{\cfbk}{\rchi(f,\bk)} 
\newcommand{\be}{\begin{equation}}
\newcommand{\ee}{\end{equation}}
\newcommand{\p}{{k_{\rm d}}}
\newcommand{\fr}{f'}
\newcommand{\fs}{f^\star_\bk}

\newcommand{\m}{{\ell}}
\newcommand{\ri}{r}
\newcommand{\falg}{{f_\m^{\rm alg}}}

\newcommand{\falgkinfi}{f^{\rm alg}_{x^\infty_r}} 
\renewcommand{\ae}{a_e}

\newcommand{\ke}{k_e}
\newcommand{\km}{k_M}
\newcommand{\A}[1]{\mathcal{A}^{#1}}
\renewcommand{\a}[1]{a^{#1}}
\newcommand{\aopt}{a_{\rm o}}
\newcommand{\mR}{\mc{R}}
\newcommand{\Gf}{G}

\newcommand{\cdotshort}{\!\cdot\!}

\newcommand{\cfk}{\rchi(f,k)}
\newcommand{\poafk}{\poa(f,k)} 
\newcommand{\poafsk}{\poa(f^\star_k,k)}
\newlength\algbrake
\newtheorem{definition}{Definition}
\newtheorem{thm}{Theorem}
\newtheorem{prop}{Proposition}
\newtheorem{lemma}{Lemma}
{\theoremstyle{remark}
\newtheorem*{rmk*}{Remark}}
\begin{document}

% title, authors, thanks
\title{\LARGE \bf
The Importance of System-Level Information in Multiagent Systems Design: Cardinality and Covering Problems
}

\author{Dario~Paccagnan$^{1}$~and~Jason~R.~Marden$^{2}$% 
\thanks{This research was supported by SNSF Grant \#P1EZP2-172122 and by AFOSR Grant \#FA9550-12-1-0359, ONR Grant \#N00014-15-1-2762, NSF Grant \#ECCS-1351866.}
\thanks{$^{1}$D. Paccagnan is with the Department of Electrical and Computer Engineering, University of California, Santa Barbara, USA and with the Automatic Control Laboratory, Swiss Federal Institute of Technology, Zurich, Switzerland. Email: {\tt dariop@control.ee.ethz.ch}.}
\thanks{$^{2}$J.\,R. Marden is with the Department of Electrical and Computer Engineering, University of California, Santa Barbara, USA. Email: {\tt jrmarden@ece.ucsb.edu.}}
\thanks{
A much abbreviated conference version of this work appeared in \cite{Allerton17Pacca}.}
}

% make the title area
\maketitle

% *** ABSTRACT ***
\begin{abstract}
A fundamental challenge in multiagent systems is to design local control algorithms to ensure a desirable collective behaviour. The information available to the agents, gathered either through communication or sensing, naturally restricts the achievable performance. Hence, it is fundamental to identify what piece of information is valuable and can be exploited to design control laws with enhanced performance guarantees. This paper studies the case when such information is uncertain or inaccessible for a class of submodular resource allocation problems termed covering problems. In the first part of this work we pinpoint a fundamental risk-reward tradeoff faced by the system operator when conditioning the control design on a valuable but uncertain piece of information, which we refer to as the cardinality, that represents the maximum number of agents that can simultaneously select any given resource. Building on this analysis, we propose a distributed algorithm that allows agents to learn the cardinality while adjusting their behaviour over time. This algorithm is proved to perform on par or better to the optimal design obtained \mbox{when the exact cardinality is known a priori}.
\end{abstract}

% *** KEYWORDS ***
\IEEEpeerreviewmaketitle

\section{Introduction}
\IEEEPARstart{S}{everal} social and engineering systems can be thought of as a collection of multiple subsystems or agents, each taking local decisions in response to available information. A central goal in this field is to design control algorithms for the individual subsystems to ensure that the collective behaviour is desirable with respect to a global objective. 
Achieving this goal is particularly challenging because of the restriction on the information available to each agent and to the large scale of typical systems. Examples include, but are not limited to, power grid networks \cite{siljak2011decentralized}, charging of electric vehicles \cite{Paccagnan16}, transportation network \cite{BrownMarden16}, task assignment problems \cite{MardenWierman13}, sensor allocation \cite{MartinezCortez07}, robotic networks \cite{PavoneArsie11}. A considerable bulk of the research has focused on the design of local control algorithms in a framework where the information at agents' disposal is itself a fixed datum of the problem. A non exhaustive list includes \cite{LangbortChandra04,BamiehPaganini02} and references therein. 
Understanding the impact of information availability on the achievable performances is a seemingly important but less tracked problem \cite{Marden14, Gairing09, LangbortGupta09}. 

Of particular interest is to recognise what supplementary piece of information could coordinate the agents to improve the system performance, and further how to incorporate this additional knowledge into a control algorithm. 
It is important to highlight that providing each agent with \emph{all} the information available to the system is in principle beneficial, but not necessarily desirable. Indeed, the communication costs associated with propagating additional information through the system might overcome the performance gains that the knowledge of additional information gives.
Therefore, the previous question has to be understood within this context. Ideally, one is interested in a piece of information that gives a significant performance enhancement, and is simple to obtain.
Loosely speaking, we measure the value of an additional piece of information with the performance gain that the best controller can offer, using that supplementary piece of information.

Relative to the class of resource allocation problems termed covering problems, \cite{Gairing09,Marden17toappear} show that the maximum number of agents that can simultaneously select a resource (which we term \emph{cardinality}) constitutes a valuable piece of information. More precisely, when the system operator is aware of the cardinality of the problem, he can devise distributed algorithms with improved performance guarantees. Nevertheless, the knowledge of the exact cardinality is in many applications uncertain, not available or may require excessive communication  to be determined. 
Following this observation, a system operator would like to understand how to operate when the knowledge of the exact cardinality is not available. What is the risk associated with using the wrong cardinality in the control design? What is the reward for using the correct one? Further and more fundamental: when the cardinality is not available at all, can the agents \emph{learn} it while simultaneously adjusting their behaviour?

The paper proceeds by considering covering problems
\cite{Vazirani01, Feige98}, a class of resource allocation problems where agents are assigned to resources in order to maximise the total value of covered items. Examples include vehicle-target assignment problems \cite{Murphey00}, sensor allocation \cite{MartinezCortez07}, task assignment \cite{FeigeVondrak06}, among others. Due to the inherent limitations in sensing and communication, in all these applications the control algorithms are required to rely only on local information.
Thus, we model distributed covering problems as strategic-form games, where the system operator has the ability to assign local objective functions to each agent.
Indeed, as shown in \cite{ChenRoughgarden10, Marden14}, Game Theory lends itself to analyse distributed systems where individual agents adjust their behaviour in response to partial information. Such game theoretic approach offers the possibility to build upon existing tools to quantify the system performance, as well as the opportunity to exploit readily available algorithms to compute equilibria in a distributed fashion \cite{MardenWierman13,Marden17toappear}. The overarching goal of the system operator is to \emph{design} local utilities in order to render the equilibria of the game as 
efficient as possible. Agents can then be guided towards an equilibrium of such game by means of existing distributed algorithms \cite{MardenWierman13, gentile2017nash}. 
It is important to highlight that we are not modelling agents as competing units, but we are rather designing their utilities to achieve the global objective.

Building on the previous results of \cite{Gairing09,Marden17toappear}, we contribute as follows.
\begin{itemize}
\item[i)] We study the problem of optimally designing the utility functions in the case when the true cardinality is not known, but only an upper bound is available.\footnote{A simple bound is given by the number of agents.} We further perform a risk-reward analysis in the case when the information on the cardinality of the game is uncertain. When the goal is to guard the system against the worst case performances, the right choice is to design the utilities as if the true cardinality was the given upper bound. Different designs will offer potential benefits, but come with a certain degree of risk. These results are presented in Theorem \ref{thm:gainsvslosses}.
\item[ii)] Motivated by the potential advantages and inherent shortcomings presented in the risk-reward analysis, we propose a distributed and asynchronous algorithm that dynamically updates the utility functions while agents adjust their behaviour over time. Such algorithm requires no initial information, and is certified to perform on par or better (in a worst case sense) to the optimal design possible, had we known the cardinality in the first place. These results are summarised in Theorem \ref{thm:alg}.
\item[iii)] We compare, instance by instance, the performance of the proposed learning algorithm with the performance of the optimal design obtained with full knowledge of the cardinality. We show that it is not possible to deem one approach superior to the other on all instances of covering problems, in that there are instances where one outperforms the other, and the converse too. These results are presented in Theorem \ref{thm:counterex}.

\end{itemize}

The remaining of the paper is organised as follows. The next section introduces the covering problem, its formulation as a strategic game and the metric used to measure the system-level performance. Section \ref{sec:unknownk} studies the utility design problem when a sole upper bound on the cardinality is available and presents the risk-reward tradeoff associated with the use of uncertain information. Section \ref{sec:dynamic} shows the possibility of dynamically adjusting the utility functions to improve the performance. Numerical simulations and conclusions follow.

\subsection*{Notation}
For any two positive integers $p \le q$, denote $[p]$ = $\{1,...,p\}$ and $[p,q] = \{p,...,q\}$; given $(\a{1},\dots,\a{n})$, denote $\a{-i}=(\a{1},\dots,\a{i-1},\a{i+1},\dots,\a{n})$. We use $\mb{N}$, $\mb{N}_0$ and $\mb{R}_{\ge0}$ to denote the set of natural numbers excluding zero,  the set of natural numbers including zero, and the set of non-negative real numbers, respectively.

\section{Distributed covering via game theory}
\label{sec:model}
In this section we present the covering problem and the associated covering game. We further define the performance metric used throughout the paper and recap previous results.
\subsection{Model}
Let us consider the problem of assigning a collection of agents $\N=\{1,\dots,n\}$ to a finite set of resources $\mc{R}=\{r_1,\dots,r_m\}$ with the goal of maximising the value of covered resources. 
The feasible allocations for each agent $i\in N$ are the elements of the action set $\a{i}\in\A{i}\subseteq 2^\mathcal R$, while every resource $r\in\mathcal{R}$ is associated with a non-negative value $v_r\ge0$. Observe that $\a{i}\subseteq \mc{R}$.  The welfare of an allocation $a=(\a{1},\dots,\a{n})\in \A{1}\times \dots\times \A{n}$ is measured by the total value of covered resources
\[W(a)\coloneqq \sum_{r\,:\, |a|_r\ge 1}v_r\,,\]
where $|a|_r$ denotes the number of agents that choose resource $r$ in allocation $a$. The \emph{covering problem} $C=\{N,\mc{R}, \{ \A{i}\}_{i\in N} ,\{v_r\}_{r\in\mc{R}}\}$ consists in finding an optimal allocation\footnote{While this problem is in general intractable, approximation algorithms for finding a near optimal solution to submodular optimization problems have been extensively studied in the literature \cite{nem78-I,krause07}. The focus of this literature is predominantly on centralized algorithms for finding near optimal allocations.  In contrast, our focus is on distributed solutions where each decision-making entity has incomplete information about the \mbox{system as a whole.}}, that is an assignment 
\[\aopt\in\arg\max_{a\in\mathcal{A}} W(a).\] Given a covering problem $C$, we define its cardinality as the maximum number of players that can concurrently select the same resource, that is
\begin{equation}
\label{eq:carddef}
\max_{r\in\mc{R},\,a\in\mc{A}}|a|_r\,.	
\end{equation}

Instead of directly specifying a distributed algorithm, we shift the focus to the design of local utility functions for each agent, as proposed first for distributed welfare games by \cite{MardenArslan09,MardenWierman13} and successively by \cite{Gairing09}. Within this framework, each agent $i\in \N$ is associated with a utility function of the form
\be
u^i(\a{i},\a{-i})\coloneqq\sum_{r\in \a{i}}v_r \cdotshort f(|a|_r)\,.
\label{eq:defutility}
\ee
  The function $f:[n]\rightarrow \mb{R}_{\ge0}$ constitutes our design choice and is called \emph{distribution rule} as it represents the fractional benefit an agent receives from each resource he selects. 
The advantages of using utilities of the form \eqref{eq:defutility} are twofold.
First, $u^i(\a{i},\a{-i})$ is local as it depends only on the resources agent $i$ selects, their value and the number of agents that selects the same resources.
Second, \eqref{eq:defutility} allows to construct a distribution rule irrespective of $\{\A{i}\}_{i\in N}$ and $\{v_r\}_{r\in\mc{R}}$ so that the final design is scalable and applies to different choices of the action sets and of the resource valuations.

Given a covering problem $C$ and a distribution rule $f:[n]\rightarrow \mb{R}_{\ge0}$, we consider the associated \emph{covering game} $\Gf\coloneqq\{C,f\}=\{N,\mc{R}, \{ \A{i}\}_{i\in N} ,\{v_r\}_{r\in\mc{R}}, f\}$, where $\A{i}$ is the set of feasible allocations and the utility of agent $i\in N$ is as in equation \eqref{eq:defutility}. 

We \emph{do not} aim at designing $f$ using information on the specific instance of covering problem at hand, as such information is often not available to the system designer. Our goal is rather to construct a distribution rule that behaves well for a large class of problems. Hence, we consider the set of covering problems for which the cardinality is exactly equal to $k\in\mb{N}$, $k\le n$. Given a distribution rule $f:[k]\rightarrow \mb{R}_{\ge0}$, we define the set of associated games as
\[
\gfk\coloneqq\{ \Gf=\{C,f\}~:\max_{r\in\mc{R},\,a\in\mc{A}}|a|_r = k\}\,.
\]
Our objective is to design $f:[k]\rightarrow \mb{R}_{\ge0}$ so that the  efficiency of all the equilibria of games in $\gfk$ is as high as possible. Note that for fixed $f$, any game $G$ is potential \cite{MardenArslan09}. Hence existence of equilibria is guaranteed and distributed algorithms, such as the best response scheme, converge to them \cite{MondererShapley96}. Throughout the paper, we focus on pure Nash equilibria \cite{Nash50}, which we will refer to in the following just as equilibria.
\begin{definition}[Pure Nash equilibrium]
\label{def:pne}
Given a game $\Gf$, an allocation $a_e\in \mc{A}$ is a pure Nash equilibrium iff $u^i(\a{i}_e,\a{-i}_e)\ge u^i(\a{i},\a{-i}_e)$ for all deviations $\a{i}\in \A{i}$ and for all players $i\in N$. In the following we use $\nashe{\Gf}$ to denote the set of Nash equilibria of $\Gf$.
\end{definition}
For a given distribution rule, we evaluate the efficiency of the Nash equilibria of games in $\gfk$, adapting the concept of Price of Anarchy from \cite{Koutsoupias99} as
\be
\label{eq:poadef}
\poafk\coloneqq\inf_{G\in \gfk}\biggl\{\frac{\min_{a\in \nashe{G}}{W(a)}}{\max_{a\in\mc{A}}W(a)}\biggr\}\le 1\,.
\ee
In essence, the quantity $\poafk$ bounds the inefficiency of the worst equilibrium (and thus of all equilibria) over games in $\gfk$, that is over games with distribution rule set to $f$ and cardinality equal to $k$.\footnote{Observe that the quantity $\poafk$ defined in \eqref{eq:poadef} bounds the price of anarchy \emph{not only} for games with cardinality exactly equal to $k$, but for all the games with cardinality smaller or equal to $k$. Indeed, among all games with cardinality smaller or equal to $k$, the worst price of anarchy is achieved in a game with cardinality exactly equal to $k$. This is because, for any game with cardinality smaller than $k$, it is possible to construct a game with cardinality $k$ that has the same price of anarchy, by assigning an additional resource $r_0$, valued zero, to $k$ agents.}
The higher the price of anarchy, the better the performance guarantees we can provide.\footnote{The quantity $W(a)$ appearing in Equation \eqref{eq:poadef} does depend on which game instance $G$ we are considering, since the resource valuations do. Hence, a more formal notation would entail using $W(a;G)$. In the interest of readability, we avoid the latter and simply use $W(a)$ when no ambiguity arise.}
\subsection{Related Work and Performance Guarantees}
The problem of designing a distribution rule so as to maximise $\poafk$ has been studied in \cite{Gairing09} and \cite{Marden17toappear}. Both works impose a natural constraint on the admissible $f$, requiring $f(1)=1$ and $f:[k]\rightarrow\mb{R}_{\ge0}$ to be non-increasing. The optimal distribution rule is explicitly derived in the former work, while the latter shows how $\poafk$ is fully characterised by a single scalar quantity $\cfk$ defined in \eqref{eq:chidef}, measuring how fast the distribution rule $f$ decreases. We intend to build upon these results, which are summarised in the following proposition. Given $k$ and a distribution rule $f$, we define 
\be
\cfk\coloneqq \max_{j\le k-1} \{j\cdot f(j)-f(j+1),(k-1)\cdot f(k)\}\,.
\label{eq:chidef}
\ee
\begin{prop}[\cite{Gairing09, Marden17toappear}]
\label{thm:optimalf}
Consider a non-increasing distribution rule $f:[k]\rightarrow \mb{R}_{\ge0}$, with $f(1)=1$.\begin{itemize}
\item[i)] The price of anarchy over the class $\gfk$ is 
\[
 \poafk=\frac{1}{1+\cfk}\,.
\]
\item[ii)] The price of anarchy over the class $\gfk$ is maximised for 
\be
f^\star_k(j)=(j-1)!\frac{\frac{1}{(k-1)(k-1)!}+\sum_{i=j}^{k-1}\frac{1}{i!}}{\frac{1}{(k-1)(k-1)!}+\sum_{i=1}^{k-1}\frac{1}{i!}}\,,\quad j\in [k]
\label{eq:optimalf}
\ee
with corresponding 
\be
\label{eq:optimalchi}
\rchi(f^\star_{k},k)=(k-1)\cdotshort f^\star_k(k)\,.
\ee
\item[iii)] The optimal price of anarchy is a decreasing function of the cardinality $k$
\be
\label{eq:optimalpoa}
\poafsk=1-\frac{1}{\frac{1}{(k-1)(k-1)!}+\sum_{i=1}^{k-1}\frac{1}{i!}}\,.
\ee
\end{itemize}
\end{prop}
\section{The case of unknown cardinality: a Risk-Reward tradeoff}
\label{sec:unknownk}
When the cardinality $k$ defining the class of games $\gfk$ is known, Proposition \ref{thm:optimalf} gives a conclusive answer on which distribution rule agents should choose to achieve the best worst case performance. In spite of that, the knowledge of the exact cardinality is in many applications not available or may require excessive communications between the agents to be determined.

Motivated by this observation, we study in the following the problem of designing a distribution rule when the cardinality $k$ defining the class of games $\gfk$ is not known, but an upper bound $k\le \bk$ is available. Observe that a universal upper bound for such quantity can be easily computed as the number $n$ of agents. Potentially tighter bounds can be derived for specific applications. 
 Our objective is to design a distribution rule $f:[\bk]\rightarrow \mb{R}_{\ge0}$ with the best performance guarantees possible with the sole knowledge of $\bk$. Once  such a distribution rule has been designed, one can use existing distributed algorithms to find an equilibrium as discussed in the introduction. Two natural questions arise in this context:
 \begin{enumerate}
 	\item  How should we select the distribution rule?
 	\item What performance can we guarantee?
 \end{enumerate}
 We will show how selecting $f^\star_{\bk}$ guards us against the worst case performance but will not guarantee the same efficiency of $f^\star_k$, when $k<\bk$. We will then present the potential benefits and risks associated with a more aggressive choice. These results motivate Section \ref{sec:dynamic}, where we will present a dynamic scheme that overcomes the difficulties encountered here, offering the same performances of $f^\star_k$ at no risk. 
 
 \subsection{Two alternative distributions}
A natural choice when an upper bound on the cardinality is available consists in designing the distribution rule exactly at the upper bound, that is using $f^\star_\bk$. A different choice might entail constructing a distribution rule where the entries $[\p]$ are designed as if the cardinality was $\p<\bk$, while the remaining entries $[\p+1,\bk]$ are optimally filled. The latter suggestion is inspired by the observation that the optimal system level performance (measured by the price of anarchy) is a decreasing function of $k$ as per \eqref{eq:optimalpoa}. This distribution is denoted with $\fr_\p$ and is constructed fixing $\fr_\p(j)=f_\p^\star(j)$ for $j\in[\p]$. The tail entries corresponding to $j\in[\p+1,\bk]$ are chosen to mitigate the risk taken. Formally, for any $1<\p<\bk$ we define the distribution rule $\fr_\p:[\bk]\rightarrow\mb{R}_{\ge0}$ as a solution of the following optimisation problem 
\be\begin{split}
\fr_\p \in \arg&\max_{f\in\mc{F}}\poa(f,\bk)\\
&\text{s.t.} \quad f(j)=f^\star_\p(j)\quad\forall j\in [\p]\,,
\end{split}
\label{eq:ftail}
\ee
where $\mc{F}=\{f:[\bk]\rightarrow\mb{R}_{\ge0} \text{ with }f(1)=1,~f(j+1)\le f(j),~\forall j\in[\bk-1]\}$ is the set of admissible distributions.
 
Note that we do not define $\fr_\p$ for $\p=1$ or $\p=\bk$ as it would reduce in both cases to $f^\star_\bk$. Further observe that the constraint $\fr_\p(j)=f^\star_\p(j),\forall j\in [\p]$ is equivalent to requiring $f'_{\p}\in\arg\max_{f\in\mc{F}}\poa(f,\p)$.

The next proposition characterises explicitly $\fr_\p$.
\begin{prop}
\label{prop:riskyf}
For any $1< \p <\bk$, the distribution $\fr_\p$ defined in \eqref{eq:ftail} is given by
\be\footnotesize
\label{eq:solutionf}
\fr_\p(j)=\begin{cases}
 f_{\p}^\star(j)\hspace*{5.1cm} j\in[\p]\\[0.1cm]
\frac{(j-1)!}{(\p-1)!}f_{\p}^\star(\p)- \rchi(\fr_\p,\bk)
\bigl( \sum_{h=1}^{j-1-\p}\frac{(j-1)!}{(j-h-1)!}+1\bigr) \quad \\\hspace*{49.5mm} j\in [\p+1,\bk]\,
\end{cases}
\ee
where $\rchi(\fr_\p,\bk)$ is given by the following expression 
\be
\label{eq:solutionchi}
\rchi(\fr_\p,\bk) = \frac{(\bk-1)(\bk-1)!}{\bk+(\bk-1)\sum_{h=1}^{\bk-1-\p}\frac{(\bk-1)!}{(\bk-h-1)!}}\frac{f_\p^\star(\p)}{(\p-1)!}.
\ee
\end{prop}
The proof is reported in Appendix \ref{app:A}.

\begin{figure}[!ht]
    \centering
    \newlength\figureheight 
	\newlength\figurewidth 
    \setlength\figureheight{3.3cm} 
	\setlength\figurewidth{6cm} 
    \definecolor{mycolor1}{rgb}{0.85000,0.32500,0.09800}%
\begin{tikzpicture}

\begin{axis}[%
width=\figurewidth,
height=\figureheight,
at={(1.011111in,0.641667in)},
scale only axis,
xmin=1,
xmax=10,
xtick={1,2,3,4,5,6,7,8,9,10},
xlabel={$j$},
ymin=0,
ymax=1,
ylabel={$f(j)$},
ylabel style = {rotate=-90},
ticklabel style = {font=\large},
]
\addplot [color=black,line width=0.5pt,mark=x,mark options={solid}, mark size=3pt]
  table[row sep=crcr]{
1   1.000000000000000\\
2   0.418023294250603\\
3   0.254069882751809\\
4   0.180232942506029\\
5   0.138955064274718\\
6   0.112798615624194\\
7   0.094814987995768\\
8   0.081728210220980\\
9   0.071848976018444\\
10  0.064664078416600\\
   };
\addlegendentry{$f^\star_{\bk}$};
\addplot [color=blue,dashed,line width=0.5pt,mark=o,mark options={solid}]
  table[row sep=crcr]{%
1	1\\
2	0.500000000000000\\
3	0.303894407615829\\
4	0.215577630463316\\
5	0.166204929469092\\
6	0.134919054961301\\
7	0.113408737383566\\
8	0.0977555693007162\\
9	0.0859389620236470\\
10	0.0773450658307411\\
};
\addlegendentry{$f'_{\p}$};
\end{axis}
\end{tikzpicture}%
    \caption{Example of distribution rules $f^\star_{\bk}$ and $\fr_{\p}$ as defined respectively in  \eqref{eq:optimalf} and \eqref{eq:solutionf}; $\bk=10$, $\p=2$.}
    \label{fig:opt-fpoa-fpos}
  \end{figure}

\begin{rmk*}
In \cite{Gairing09} a distribution rule $f$ was required to satisfy the constraint $j\cdotshort f(j)\le 1$ for all $j$. Loosely speaking the above requirement guarantees that a distribution rule does not overpay the players, in the sense that $\sum_{i\in N} u^i(\a{i},\a{-i})\le W(a)$ for all allocations.
Observe that such constraint might be important for economic applications, but it is irrelevant in the design of engineering systems. 
While \cite{Gairing09} shows that the distribution rule $f^\star_{\bk}$ satisfies this property, the next lemma proves that also $\fr_{\p}$ verifies this condition even if this was not requested a priori.
\end{rmk*}

\begin{lemma}
\label{lemma:riskydoesnotoverpay}
For any $1< \p< \bk$ the distribution $\fr_\p:[\bk]\rightarrow\mb{R}_{\ge0}$ satisfies $j\cdotshort \fr_\p(j)\le1$ for all $j\in[\bk]$.
\end{lemma}
The proof is provided in Appendix \ref{app:A}.
\subsection{Performance comparison}
Based on the metric introduced in \eqref{eq:poadef}, we compare in this section the performance of $f^\star_\bk$ with the performance of $\fr_\p$. Theorem \ref{thm:gainsvslosses} constitutes the main result of this section. 
\begin{thm}
\label{thm:gainsvslosses}
Consider the set of games $\gfk$, where $k\le  \bk$.
\begin{itemize}
\item [i)] The distribution $f^\star_{\bk}$ has performance
\[
\poa(f^{\star}_{\bk},k)=\poa(f^{\star}_{\bk},\bk)\,.
\]
Such performance is strictly worse than the one achieved by the distribution $f^\star_k$ if $k<\bk$ and equal if $k=\bk$. 
\item[ii)] For $1 < \p < k$ the distribution $\fr_\p$  has performance
\[
\poa(\fr_\p,k)=\poa(\fr_\p,\bk)\,,
\]
which is strictly worse than the one achieved by $f^\star_\bk$.  

For $k\le\p <\bk$ the distribution $\fr_\p$ has performance
\[
\poa(\fr_\p,k)=\poa(f^\star_\p,\p) \,,
\]
which is strictly better than the one achieved by $f^\star_\bk$. 
\end{itemize}
\end{thm}
The proof can be found in Appendix \ref{app:A}. 

\begin{rmk*}
Claim i) in Theorem \ref{thm:gainsvslosses} shows that the performance of the distribution $f^\star_\bk$ on the class of games with cardinality equal to $k$
is independent on the actual value of $k$, as for any $k\le \bk$, such performance is governed by $\poa(f^\star_\bk,\bk)$. Claim ii) in Theorem \ref{thm:gainsvslosses} ensures that the distribution $\fr_\p$ 
outperforms $f^\star_\bk$ for $\bk>\p\ge k$ and the opposite when $\p< k$.
In each of these cases the performance is independent on the actual value of $k$, but only depends on whether $\p$ is above or below $k$.
Loosely speaking, if we underestimate $k$ by designing $\p<k$, the performance guarantees offered by $\fr_\p$ are worse than what $f^\star_\bk$ can achieve. The reverse holds in the case when we overestimate the cardinality as in $\bk>\p\ge k$.  
\end{rmk*}

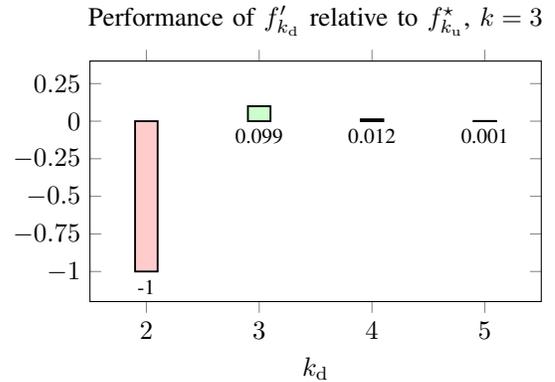
\begin{figure}[h!] 
\begin{center} 
\setlength\figureheight{3.2cm} 
\setlength\figurewidth{6cm} 
\definecolor{mycolor1}{rgb}{0.00000,0.44700,0.74100}%

\pgfplotsset{
        compat=1.7,
        % define your own legend style here
        my ybar legend/.style={
            legend image code/.code={
                \draw [##1] (0cm,-0.6ex) rectangle +(2em,1.5ex);
            },
        },
    }

\begin{tikzpicture}

\begin{axis}[%
ybar,
bar width = 0.3cm,
width=\figurewidth,
height=\figureheight,
scale only axis,
xmin=1.5,
xmax=5.5,
xtick={2,3,4,5},
xlabel={$\p$},
ymin=-1.2,
ymax=0.4,
yminorticks=true,
ytick={-1,-0.75,-0.5,-0.25,0,0.25},
title={Performance of $\fr_\p$ relative to $f^\star_\bk$, $k=3$},
legend pos=north east,
my ybar legend,
bar shift = 0pt,
]

\addplot [color=black,fill=red!20,solid,line width=.7pt, mark options={solid}]
  table[row sep=crcr]{%
2	-1\\
};

\addplot [color=black,fill=green!20,solid,line width=.7pt, mark options={solid}]
  table[row sep=crcr]{%
3	0.0997558358386140\\
4	0.0125192580420671\\
5	0.00148934590687065 \\
};

\node[] at (50,10) { \footnotesize -1};
\node[] at (150,110) { \footnotesize 0.099};
\node[] at (250,110) { \footnotesize 0.012};
\node[] at (350,110) { \footnotesize 0.001};

\end{axis}
\end{tikzpicture}%  
\caption{The bars represent the difference $\poa(\fr_\p,k)-\poa(f^\star_\bk,k)$, normalized by its largest value. As such, it describes the normalized difference in performance between $\fr_\p$ and $f^\star_\bk$ for various values of $2\le\p\le5$, $k=3$, $\bk=6$.}
\label{fig:gainsvslosses2}
\end{center}
\end{figure}

In Figure \ref{fig:gainsvslosses2} we compare the performance of $\fr_\p$ with the performance of $f^\star_\bk$. 
It is important to note that the performance degradation (incurred whenever  $\p<k$) always \emph{dominates} the potential gains (achieved when $\p\ge k$). This is also exemplified in Table \ref{tb:performance} and motivates the next section where we will introduce a dynamic algorithm capable of offering the benefits of $f^\star_k$ without the knowledge of $k$.

\begin{table}[h!]
\begin{center}
\caption{Comparison between $\fr_\p$ and $f^\star_{k_u}$ for $2\le k_d\le 5$, $k=3$, $k_u=6$.}\label{tb:performance}
\begin{tabular}{cc}
$\p$ & $\frac{\poa(\fr_\p,k)-\poa(f^\star_{k_u},k)}{\poa(f^\star_{k_u},k)}$ in $\%$ \\
\hline\\[-1.5mm]
\hspace*{-1.5mm}   
$2$ & $-6.727$ \\[0.5mm]
$3$ & $~~0.670$ \\ [0.5mm]
$4$ & $~~0.083$ \\ [0.5mm]
$5$ & $~~0.009$ \\ \hline
\end{tabular}
\end{center}
\end{table}
%
%\vspace*{-5mm}
\section{Beyond the Risk-Reward tradeoff}
\label{sec:dynamic}
The previous section has focused on the design of a distribution rule when an upper bound on the true cardinality is known. We have demonstrated how $f^\star_\bk$ guards against worst case performance while $\fr_\p$ could give potential benefits, but comes with a certain degree of risk. In both cases the performance is equal or inferior to what we could achieve if we knew the true cardinality.

In this section we show how to overcome such difficulties, when we are given a game $G\in\gfk$ with \emph{unknown} cardinality $k$. We propose a distributed and asynchronous implementation of the best response algorithm that dynamically updates which distribution rule to use. The upshot is that we guarantee an equal or superior performance to what we could achieve if we knew $k$.

In the following, we allow distribution rules to depend on an additional variable $x_r\in [n]$ defined for $r\in\mc{R}$, which we will dynamically update to coordinate the agents. In particular, we generalise the utilities of \eqref{eq:defutility} to 
\be
\label{eq:defutilityspecific}
u^i(\a{i},\a{-i};x)\coloneqq\sum_{r\in \a{i}}v_r f(x_r,|a|_r)\,,
\ee
where $x=\{x_r\}_{r\in\mc{R}}$ and $f:[n]\times[n]\rightarrow \mb{R}_{\ge0}$ might be different across the resources, depending on the value of $x_r$. 

One could question whether the improved performance we will obtain comes from the additional degree of freedom introduced allowing resource specific distribution rules. Nevertheless \cite{Gairing09} shows that it is not the case, in that the best resource specific and non resource specific distribution perform equally (in the worst case sense). The only rationale to introduce resource dependent rules is the distributability of the algorithm. Indeed, similar results could have been achieved dynamically updating a single distribution rule shared by all resources, but such algorithm would have not been distributed.

\subsection{Algorithm description and distributedness}
In the following $t\in\mb{N}_0$ describes the time step of the algorithm and $a_t\in\mc{A}$ the corresponding allocation. With slight abuse of notation,
for every resource $r\in \mc{R}$ we introduce the quantity $x_r(t)$ that associates $r\in\mc{R}$ to the maximum number of agents that chose such resource until time $t\in\mb{N}_0$. Further, we define $\falg:[n]\rightarrow\mb{R}_{\ge0}$ for every $\m\in\mb{N}$ as a distribution rule\footnote{The rule $\falg$ is a valid distribution rule, being non increasing and such that $\falg(1)=1$. It will in general not satisfy $j\cdot\falg(j)\le 1$, but this was neither requested, nor has relevance in the design of engineering systems.} matching the optimal in equation \eqref{eq:optimalf} for $j\in[\m]$ and constant in between $[\m,n]$
\be
\falg(j)\coloneqq\begin{cases}
f^\star_\m(j)\hspace*{3mm} j\in[\m]\,,\\
f^\star_\m(\m)\quad j\in[\m+1,n]\,.
\end{cases}
\label{eq:falg}
\ee
\begin{algorithm}[h!]
\caption{Asynchronous cardinality learning}\label{alg:learning}
\begin{algorithmic}[1]
\State Initialise $a_0$;\quad $t\gets 0$;\quad$x_\ri(t)\gets|a_0|_r$~~\hfill$\forall \ri\in\mR$
\vspace*{\algbrake}\vspace*{\algbrake}
\While {not converged}
\vspace*{\algbrake}
\LeftComment{Best response}
\State $i\gets 1 +  (t\mod n)$
\State $\a{i}_{t+1}\gets \arg\max_{\a{i}\in\mathcal{A}_i} \sum_{r\in \a{i}}v_r f^{\rm alg}_{x_\ri(t)}(|a_t|_r)$
\State  $\a{i}_{t+1}\gets (\a{i}_{t+1},\a{-i}_t)$

\vspace*{\algbrake}
\LeftComment{Update $k_t$ and thus $f$ on every resource}
\State $x_\ri(t+1) \gets \max\{x_\ri(t) ,{|a_{t+1}|_{\ri}}\}$\hfill$\forall \ri\in\mR$

\State $t\gets t+1$
\EndWhile
\end{algorithmic}
\end{algorithm}

\vspace*{3mm}
Through the additional variable $x_\ri(t)$, the algorithm keeps track of the maximum number of players that visited every resource until the current time $t$, and selects consequently a resource specific distribution rule. In particular on every $r\in\mc{R}$, the algorithm uses $\falg$ with $\ell$ set to the maximum number of players that visited that resource until time $t$ (lines $4$ and $6$).
Following a round-robin rotation, players $i$ is selected to best respond and update the allocation (lines $3$ to $5$). The procedure repeats until convergence.\footnote{{Note that the best response strategy is not guaranteed to be unique. To overcome this issue, in the following we assume the existence of a tie-breaking rule selecting a single optimal allocation, should these be multiple. 
Nevertheless, we observe that neither this, nor requiring players to best respond in a round-robin is fundamental. It is still possible to show that Algorithm \ref{alg:learning} converges almost surely if the players best responding are uniformly randomly selected form $[n]$ and a single optimal allocation is uniformly randomly extracted from the set of best responses. This will produce a totally asynchronous algorithm.}}

The algorithm is distributed in the sense that every agent needs to keep track of $x_r(t)$ only for those resources he has access to i.e. for $r\in \A{i}$. Further, it is asynchronous as players need not to update their allocation in a specified order, but can spontaneously revise their strategies (see footnote 6). It is important to highlight that the communication requirements of Algorithm \ref{alg:learning} are the same of those needed by the best response algorithm applied for instance to distribution rules $f^\star_\bk$ or $\fr_\p$. That is, Algorithm \ref{alg:learning} better exploits the information that is already available.

In Figure \ref{fig:secondcomparison} we compare the distribution $f^\star_k$ with $f^{\rm alg}_{x_r^\infty}$, where $x_r^\infty=\lim_{t\to\infty} x_r(t)$. We exemplify such comparison for some of the allowed values $1\le x_r^\infty\le k$.
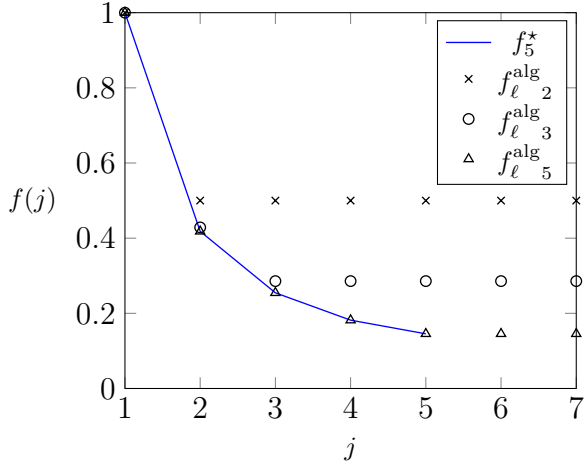
\begin{figure}[!ht]
    \centering
    \setlength\figureheight{5cm} 
	\setlength\figurewidth{6cm} 
    \definecolor{mycolor1}{rgb}{0.85000,0.32500,0.09800}%
\begin{tikzpicture}

\begin{axis}[%
width=\figurewidth,
height=\figureheight,
at={(1.011111in,0.641667in)},
scale only axis,
xmin=1,
xmax=7,
xtick={1,2,3,4,5,6,7},
xlabel={$j$},
ymin=0,
ymax=1,
ylabel={$f(j)$},
ylabel style = {rotate=-90},
ticklabel style = {font=\large},
]
\addplot [color=blue,line width=0.5pt,]
  table[row sep=crcr]{
1   1.000000000000000\\
2   0.418181818181818\\
3   0.254545454545455\\
4   0.181818181818182\\
5   0.145454545454545\\
   };
\addlegendentry{$f^\star_{5}$};
\addplot [only marks, color=black,line width=0.5pt,mark=x,mark options={solid}]
  table[row sep=crcr]{
1   1.000000000000000\\
2   0.5\\
3   0.5\\
4   0.5\\
5   0.5\\
6   0.5\\
7   0.5\\
   };
\addlegendentry{$\falg_{2}$};
\addplot [only marks, color=black,line width=0.5pt,mark=o,mark options={solid}]
  table[row sep=crcr]{
1   1.000000000000000\\
2   0.428571428571429\\
3   0.285714285714286\\
4   0.285714285714286\\
5   0.285714285714286\\
6   0.285714285714286\\
7   0.285714285714286\\
   };
\addlegendentry{$\falg_{3}$};
\addplot [only marks, color=black,line width=0.5pt,mark=triangle,mark options={solid}]
  table[row sep=crcr]{
1   1.000000000000000\\
2   0.418181818181818\\
3   0.254545454545455\\
4   0.181818181818182\\
5   0.145454545454545\\
6   0.145454545454545\\
7   0.145454545454545\\
   };
\addlegendentry{$\falg_{5}$};
\end{axis}
\end{tikzpicture}%
    \caption{Comparison between $f^\star_k$ and $f^{\rm alg}_{x_r^\infty}$, in the case of $k=5$, $n=7$, for different values $x_r^\infty=\{2, 3, 5\}$.}
    \label{fig:secondcomparison}
  \end{figure}
  
\newpage  
For ease of exposition we have presented the case where the distribution rules depend on the history $x_r(t)$, but the same across the players. It is simple to extend these results to the case of agent specific distribution rules. Every player would use a resource specific distribution rules that depend on the maximum number of players that visited every resource up until \emph{his} last visit. Similar convergence guarantees and performance certificates will follow.
\subsection{Convergence and quality of equilibria}
The following theorem is the main result of this section.  Claim i) shows convergence of Algorithm \ref{alg:learning} to a Nash equilibrium. Claim ii) proves that the quality of such equilibrium is higher or equal to what the optimal distribution $f^\star_k$ could achieve. 
\begin{thm} Consider a covering game $G$ with cardinality $k$.
\label{thm:alg}
\begin{itemize}
\item[i)]
Algorithm \ref{alg:learning} converges in a finite number of steps to $\ae\coloneqq\lim_{t\to\infty}a_t$ $\in\mc{A}$. The allocation $\ae$ is a Nash equilibrium of the game with resource specific distribution rules fixed to $f^{\rm alg}_{x_\ri^\infty}$ for $\ri\in\mR$, where $x_\ri^\infty\coloneqq\lim_{t\to\infty} x_\ri(t)$.\\
\item[ii)] Let $\km\coloneqq \max_{\ri\in\mR}x_\ri^\infty$. The quality of the equilibrium $\ae$ is higher than $\poa({f^\star_{\km}},{\km})$ and thus of $\poa({f^\star_{k}},{k})$ 
\be\frac{W(\ae)}{W(\aopt)}\ge \poa({f^\star_{\km}},{\km})\ge\poa({f^\star_{k}},{k})\,,
\ee
where $\aopt\in\arg\max_{a\in\mathcal{A}} W(a)$.
\end{itemize}
\end{thm}
These statements hold for any initial condition, even if the allocation to which the Algorithm \ref{alg:learning} converges may be different.
The proof is detailed in Appendix \ref{app:A}.
\begin{rmk*}
The reason for which the proposed algorithm gives a performance that is on par or better compared to what offered by $f^\star_k$ is, informally, in the structure of the equilibria induced by $f^{\rm alg}_{x_r^\infty}$. More precisely, i) for each equilibrium $a_e$, the number of agents selecting resource $r$ is $|a_e|_r\le x_r^\infty \le k$; and ii) at convergence resource $r$ is using $f^{\rm alg}_{x_r^\infty}$ (with $x_r^\infty\le k$), which exactly matches the optimal $f^\star_{x_r^\infty}$ for $j\le x_r^\infty$. This is enough to guarantee an improved performance. The proof of Theorem \ref{thm:alg} makes this reasoning formal.
\end{rmk*}
\subsection{Instance by instance analysis}
The previous theorem shows that Algorithm \ref{alg:learning} achieves a higher or equal worst case performance than the optimal distribution $f^\star_k$.
While worst case analysis has been and still is a fruitful tool to measure and improve on algorithms' performance, the computer science community has recently showed interest in moving beyond it \cite{BalcanManthey15}.
Inspired by this, the question arises as to whether Algorithm \ref{alg:learning} performs better than $f^\star_k$, instance by instance.
More formally, we would like to understand if Algorithm \ref{alg:learning} yields higher welfare than the optimally designed rule on all the remaining instances (the non worst case ones). We show that neither this, \mbox{nor the converse holds.}
 
\begin{thm}
\label{thm:counterex}
Let $C$ be an instance of covering problem defined in Section \ref{sec:model}. Further denote with $\nashealg{C}$ the set of equilibria obtained using Algorithm \ref{alg:learning} on $C$, and $G^\star=\{C,f^\star_k\}$ the associated game where the optimal distribution $f^\star_k$ has been selected.
\begin{itemize} 
	\item[]
	\item[i)] There exists an instance $C$ of the covering problem such that 
	\[
	\min_{a\in \nashealg{C}}W(a)>	\min_{a \in \nashe{G^\star}}W(a)\,.
	\]
\item[ii)] There exists an instance $C$ of the covering problem such that 
	\[
	\min_{a\in \nashealg{C}}W(a)<	\min_{a \in \nashe{G^\star}}W(a)\,.
	\]
\end{itemize}

\end{thm}
The proof is constructive and is presented in Appendix \ref{app:A}.
Note that both statements in Theorem \ref{thm:counterex} compare the performances of a given covering problem $C$ and associated game $G^\star$. Observe that this metric is significantly different from \eqref{eq:poadef}, where we additionally take the infimum over problems with cardinality equal to $k$.
\section{Simulations}
\label{sec:simulations}
In this section we provide simulations to compare the performance of different distribution rules. 

For this numerical study, we consider the problem of distributed data caching introduced in \cite{goemans2006market} as a technique to reduce peak traffic in mobile data networks. In order to alleviate the growing radio congestion caused by the recent surge of mobile data traffic, the latter work suggested to store popular and spectrum intensive items (such as movies or songs) in geographically distributed stations. The question we seek to answer is how to distribute the popular items across the nodes of a network so as to maximize the total number of queries fulfilled.
More in details, we consider a square grid of $800 \times 800$ bins and a set of geographically distributed agents $N$ (the stations) with position $P_i$. Additionally, we consider a set $\mc{R}$ of data items, where each item has query rate $q_r$ and is also geographically distributed with position $O_r$. The allocation set of agents $i$ is $\mc{A}_i\subseteq 2^{\mc{R}_i}$, where $\mc{R}_i$ are the set of resources $r\in\mc{R}$ that are sufficiently close to the considered station. Formally,  $\mc{R}_i=\{r\in\mc{R}\,s.t.\,||O_r-P_i||\le \rho\}$. In addition, we require that $|a_i|\le l_i$ for all $a_i\in\mc{A}_i$ as the storage capacity is limited in each station. The goal is to select a feasible allocation on every station so as to jointly maximize the number of queries fulfilled i.e.
\[
W(a)=\sum_{r\,:\,|a|_r\ge 1} q_r.
\] 
In the following we consider $|N|=150$, $|\mc{R}|=1500$, $\rho=50$, $l_i=10$ for all $i$. Data items are randomly located in the grid (with uniform distribution), while the corresponding query rates follow the Zipf distribution\footnote{Query rates {approximately follow this distribution, as shown in \cite{breslau1999web}.}} $q_r = 1/r^\alpha$ with $\alpha = 0.6$.
 The stations are uniformly distributed, on the grid. All the instances considered have cardinality $k=3$.
 We considered $5\cdot 10^4$ instances of this problem, and for every instance we computed an equilibrium allocation using the best response algorithm in conjunction with $f'_{k_d}$, $f^\star_\bk$, $f^\star_k$ or Algorithm \ref{alg:learning}. Given the size of the problem, it is not possible to compute the optimal allocation and thus the price of anarchy. As a surrogate for the latter we use the ratio between $W(\ae)$ and $W_{\rm tot}$, where $W_{\rm tot}=\sum_{r \in \mc{R}} q_r$ is the total value of queries.  This ratio provides a lower bound for the true price of anarchy as $W(\aopt)\le W_{\rm tot}$. Observe that $W_{\rm tot}$ is constant throughout any instance considered as $W_{\rm tot}=\sum_{r\le 1500} 1/r^\alpha$. Thus, it is possible to compare the performance across different instances by looking at $W(\ae)/W_{\rm tot}$.

In Figure \ref{fig:comparisonpdf} (top) we compare the empirical distribution of the ratio $W(\ae)/W_{\rm tot}$ for the rules that use no information about the true cardinality i.e. $f'_{k_d}$, $f^\star_\bk$ for $k_d=2$ and $\bk=5$. 
In Figure \ref{fig:comparisonpdf} (bottom) we compare the performance of the worst case optimal distribution $f^\star_k$ with that of our learning Algorithm \ref{alg:learning}.
The worst case ratio $W(\ae)/W_{\rm tot}$ ever encountered for each case is represented in Figure \ref{fig:comparisonpdf} with a marker, and is also reported in the following table. Additionally, in Table \ref{tb:compperformance} we show the maximum, minimum, and average number of best response rounds.\footnote{Observe that in each best response round \emph{all} the agents have a chance to update their allocation.}

\begin{figure}[h!] 
\begin{center} 
\setlength\figureheight{3.2cm} 
\setlength\figurewidth{6cm} 
\begin{tikzpicture}

\begin{axis}[%
width=\figurewidth,
height=\figureheight,,
at={(1.011in,0.642in)},
scale only axis,
xmin=0.88,
xmax=0.95,
xlabel={$W(\ae)/W_{\rm tot}$},
ymin=-300,
ymax=5000,
unbounded coords=jump,
scaled y ticks=base 10:-3,
ylabel={Count},
ylabel style={rotate=-90},
grid=both,
legend style={at={(0.05,0.8)}, anchor=west, draw=white!15!black},
legend cell align={left}
]
\addplot[color=green, mark = triangle, mark size = 3pt, mark options={solid}, line width = 1pt]
  table[row sep=crcr]
  {0.8823	1\\};
\addlegendentry{$f'_{k_d}$\hspace*{-1mm}};
\addplot[color=red, mark = x, mark size = 3pt, line width = 1pt]
  table[row sep=crcr]{%
0.9052	0\\};
\addlegendentry{$f^\star_{k_u}$};
\addplot[ybar interval, fill=green, fill opacity=0.6, draw=black, area legend] table[row sep=crcr] {%
x	y\\
0.8823  1\\
0.8994	2\\
0.90005	4\\
0.9007	7\\
0.90135	15\\
0.902	23\\
0.90265	54\\
0.9033	101\\
0.90395	175\\
0.9046	275\\
0.90525	401\\
0.9059	656\\
0.90655	938\\
0.9072	1383\\
0.90785	1930\\
0.9085	2398\\
0.90915	2811\\
0.9098	3357\\
0.91045	3823\\
0.9111	4176\\
0.91175	4250\\
0.9124	4081\\
0.91305	3839\\
0.9137	3440\\
0.91435	2999\\
0.915	2373\\
0.91565	1871\\
0.9163	1478\\
0.91695	1093\\
0.9176	795\\
0.91825	497\\
0.9189	332\\
0.91955	198\\
0.9202	111\\
0.92085	61\\
0.9215	26\\
0.92215	16\\
0.9228	8\\
0.92345	1\\
0.9241	1\\
0.92475	1\\
0.9254	1\\
};
\addplot[ybar interval, fill=red, fill opacity=0.6, draw=black, area legend] table[row sep=crcr] {%
x	y\\
0.9052  1\\
0.9186	1\\
0.91986	2\\
0.92049	2\\
0.92112	9\\
0.92175	8\\
0.92238	19\\
0.92301	38\\
0.92364	57\\
0.92427	77\\
0.9249	211\\
0.92553	323\\
0.92616	524\\
0.92679	778\\
0.92742	1151\\
0.92805	1649\\
0.92868	2257\\
0.92931	2820\\
0.92994	3391\\
0.93057	3822\\
0.9312	4255\\
0.93183	4444\\
0.93246	4311\\
0.93309	4222\\
0.93372	3746\\
0.93435	3089\\
0.93498	2580\\
0.93561	2016\\
0.93624	1485\\
0.93687	992\\
0.9375	709\\
0.93813	451\\
0.93876	224\\
0.93939	160\\
0.94002	91\\
0.94065	41\\
0.94128	27\\
0.94191	13\\
0.94254	3\\
0.94317	2\\
0.9438	2\\
};
\end{axis}
\end{tikzpicture}%  
\begin{tikzpicture}

\begin{axis}[%
width=\figurewidth,
height=\figureheight,,
at={(1.011in,0.642in)},
scale only axis,
xmin=0.88,
xmax=0.95,
xlabel={$W(\ae)/W_{\rm tot}$},
ymin=-300,
ymax=5000,
unbounded coords=jump,
scaled y ticks=base 10:-3,
ylabel={Count},
ylabel style={rotate=-90},
grid=both,
legend style={at={(0.05,0.8)}, anchor=west, draw=white!15!black},
legend cell align={left}
]
\addplot[color=black, mark = asterisk, mark size = 3pt, mark options={solid}, line width = 1pt]
  table[row sep=crcr]
  {0.9125	1\\};
\addlegendentry{$f^\star_k$\hspace*{-1mm}};
\addplot[color=magenta, mark = o, mark size = 3pt, line width = 1pt]
  table[row sep=crcr]{%
0.9186	0\\};
\addlegendentry{Alg. \ref{alg:learning}};
\addplot[ybar interval, fill=black, fill opacity=0.6, draw=black, area legend] table[row sep=crcr] {%
x	y\\
0.9125	1\\
0.9181	1\\
0.9187	3\\
0.9193	12\\
0.9199	14\\
0.9205	28\\
0.9211	45\\
0.9217	64\\
0.9223	152\\
0.9229	215\\
0.9235	366\\
0.9241	551\\
0.9247	842\\
0.9253	1155\\
0.9259	1646\\
0.9265	2114\\
0.9271	2743\\
0.9277	3162\\
0.9283	3755\\
0.9289	3909\\
0.9295	4126\\
0.9301	4005\\
0.9307	4035\\
0.9313	3684\\
0.9319	3258\\
0.9325	2686\\
0.9331	2164\\
0.9337	1668\\
0.9343	1249\\
0.9349	829\\
0.9355	574\\
0.9361	395\\
0.9367	252\\
0.9373	139\\
0.9379	77\\
0.9385	44\\
0.9391	17\\
0.9397	12\\
0.9403	6\\
0.9409	2\\
0.9415	2\\
};
\addplot[ybar interval, fill=magenta, fill opacity=0.6, draw=black, area legend] table[row sep=crcr] {%
x	y\\
0.9186	1\\
0.91923	0\\
0.91986	1\\
0.92049	3\\
0.92112	5\\
0.92175	7\\
0.92238	15\\
0.92301	27\\
0.92364	38\\
0.92427	65\\
0.9249	149\\
0.92553	257\\
0.92616	425\\
0.92679	651\\
0.92742	962\\
0.92805	1434\\
0.92868	1970\\
0.92931	2533\\
0.92994	3089\\
0.93057	3667\\
0.9312	4052\\
0.93183	4403\\
0.93246	4302\\
0.93309	4411\\
0.93372	3967\\
0.93435	3396\\
0.93498	2832\\
0.93561	2313\\
0.93624	1682\\
0.93687	1183\\
0.9375	871\\
0.93813	553\\
0.93876	307\\
0.93939	212\\
0.94002	105\\
0.94065	52\\
0.94128	36\\
0.94191	14\\
0.94254	8\\
0.94317	2\\
0.9438	2\\
};
\end{axis}
\end{tikzpicture}% 
\caption{Empirical distribution of $W(\ae)/W_{\rm tot}$ on $5\cdot 10^4$ instances of the covering problem considered; $k=3$, $k_d=2$ and $\bk=5$.}
\label{fig:comparisonpdf}
\end{center}
\end{figure}
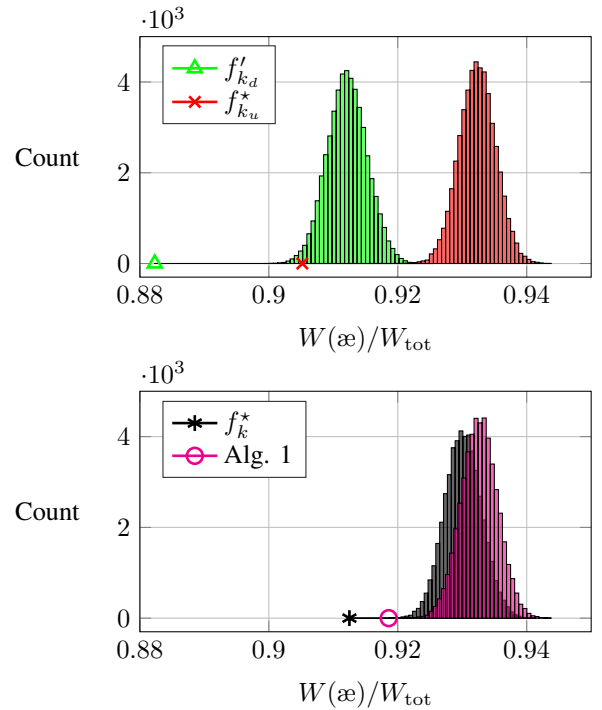

\begin{table}[hb]
\begin{center}
\caption{Performance Comparison}\label{tb:compperformance}
\begin{tabular}{ccccc}
Algorithm & Min $W(\ae)/W_{\rm tot}$ & Min \#BR & Max \#BR & Avg \#BR\\
\hline\\[-1.5mm]
BR with $f'_{k_d}$ & 0.8823 & 2 & 5 & 3.32 \\[0.5mm]
BR with $f^\star_\bk$ & 0.9052 & 3 & 5 & 3.28\\[0.5mm]
\!\!\!\! BR with $f^\star_k$ & 0.9125  & 2 & 5 & 3.23\\ [0.5mm]
Algorithm \ref{alg:learning} & 0.9186 & 3 & 5 & 3.29  \\ \hline
\end{tabular}
\end{center}
\end{table}

First, we note that all the tested algorithms require a comparable number of best response rounds, and thus have very similar running time.
Second, we observe that $f'_{k_d}$ performs the worst among all the other distributions, both in terms of wort-case performance, and in terms of average performance.
Additionally, we note that Algorithm \ref{alg:learning} and the distribution rules $f^\star_\bk$, $f^\star_k$ perform similarly, when looking at an average instance, while $f^\star_k$ and Algorithm \ref{alg:learning} outperform $f^\star_\bk$ in terms of worst case performance with a slight advantage for Algorithm \ref{alg:learning}. The efficiency values are much higher compared to the analytical worst case, hinting at the fact that such instances are very few.
Given that the average performance is similar, but the distribution $f^\star_\bk$ is proven to have inferior worst case performance (Theorem \ref{thm:gainsvslosses}), one might want to use either the optimal distribution $f^\star_k$ or Algorithm \ref{alg:learning}. Recall indeed that the worst case performance of Algorithm \ref{alg:learning} is on par or better to $f^\star_k$ (Theorem \ref{thm:alg}). Nevertheless, the use of $f^\star_k$ requires knowledge of the cardinality $k$, while the algorithm proposed does not.

To conclude: Algorithm \ref{alg:learning} achieves similar average performances compared to $f^\star_k$, but has a better worst case performance than $f^\star_\bk$ and a better-equal worst case performance than $f^\star_k$ even if it does not require the knowledge of $k$.

\section{Conclusion}
In this work we studied how additional information impacts the optimal design of local utility functions, when the goal is to improve the overall efficiency of a class of multiagent systems. Relative to covering problems, in the first part of the manuscript we highlighted an inherent tradeoff between potential risks and rewards when such additional information is uncertain. In the second part, we showed how it is possible to fully eliminate the risks by using a distributed algorithm that dynamically updated the local utility functions.
	The methodology used suggests that similar results could be obtained for a broader class of resource allocation problems than the one studied here.

\appendices
\section{}
\label{app:A}
In the proofs presented in Appendix \ref{app:A}, we make use of Lemmas \ref{lemma:generalinequality}-\ref{lemma:keypoints}. Their presentation and proof is deferred to Appendix \ref{app:lemmata}.
\section*{Proof of Proposition \ref{prop:riskyf}}
Thanks to result i) in Proposition \ref{thm:optimalf} maximising $\poa(f,\bk)$ is equivalent to minimising $\cfbk$ and $\fr_\p$ can be computed by the following linear program (LP) in the unknowns $x$, $\{f(j)\}_{j=1}^{\bk}$
\be
\begin{split}
&\min_{x\ge0,\,f\in\mc{F}} ~x\\
&\text{s.t.}~ jf(j)-f(j+1)\le x\quad j \in [\bk-1]\,,\\
&\hspace*{5mm}(\bk-1)f(\bk)\le x\,,\\
&\hspace*{5mm}f(j) = f^\star_\p(j)\qquad  \hspace*{12mm}j\in[\p]\,.
\end{split}
\label{eq:original}
\ee
We remove the constraints $x\ge0$, $f\in\mc{F}$ as well as $ jf(j)-f(j+1)\le x$ for $j \in [\p-1]$ and introduce the following relaxed linear program
\be
\begin{split}
&\min_{x,\,f} ~x\\
&\text{s.t.}~ jf(j)-f(j+1)\le x\quad j \in [\p,\bk-1]\,,\\
&\hspace*{5mm}(\bk-1)f(\bk)\le x\,,\\
&\hspace*{5mm}f(j) = f^\star_\p(j)\qquad  \hspace*{12mm}j\in[\p]\,.
\end{split}
\label{eq:relaxed}
\ee
The proof is divided in two subproofs:\begin{itemize}
\item[i)] We show that a solution to the relaxed program \eqref{eq:relaxed} is given by \eqref{eq:solutionf} and \eqref{eq:solutionchi}.
\item [ii)] We show that the solution to the relaxed program obtained in i) is feasible for the original problem too.
\end{itemize}
\begin{proof} i) The proof proceeds by showing that a solution of \eqref{eq:relaxed} can be obtained transforming all the inequality constraint into equalities. This will produce the expressions \eqref{eq:solutionf} and \eqref{eq:solutionchi}.

Let us define $v_j=f(j)$ for $j\in[\p+1,\bk]$ and introduce the cost function $J(x,v_{\p+1},\dots,v_{\bk})=x$. We further introduce the constraint functions $g_1(x,v_{\p+1})=-x-v_{\p+1}$ and
$g_i(x,v_{p+i-1},v_{p+i})=-x+j\, v_{p+i-1}-v_{p+i}$ for $i\in[2,\bk-p]$ and  $g_{\bk-\p+1}(x,v_{\bk})=-x+(\bk-1)v_{\bk}$.
With these definitions the LP \eqref{eq:relaxed} is equivalent to the following where we have removed the decision variables that are already determined
\[
\begin{split}
&\min_{x,v_{\p+1},\dots,v_\bk} ~J(x,v_{\p+1},\dots,v_\bk)\,,\\
&\hspace*{6mm}\text{s.t.\quad } g_1(x,v_{\p+1})\le -\p f_\p^\star(\p)\,, \\
&\hspace*{13.5mm} g_i(x,v_{p+i-1},v_{p+i})\le0\qquad i\in[2,\bk-p]\,,\\
&\hspace*{13.5mm} g_{\bk-\p+1}(x,v_{\bk})\le0\,.
\end{split}
\]
Thanks to the convexity of the cost function and to the polytopic constraints, the Karush-Kuhn-Tucker conditions are necessary and sufficient for optimality \cite{boyd2004}. Consequently, a feasible point $(x^\star,v_{k+1}^\star,\dots,v_n^\star)$ is an optimiser iff there exists $\mu_i\in\mb{R}$ so that
\[\begin{split}
&\nabla J^\star+\sum_{i=1}^{\bk-\p+1}\mu_i\nabla g_i^\star=0\\
&g_i^\star\le 0, \quad \mu_i\ge0, \quad \mu_i g_i^\star = 0 \quad i\in[\bk-\p+1]
\end{split}
\]
where we used $\nabla J^\star$ to indicate $\nabla J(x^\star,v_{\p+1}^\star,\dots,v_\bk^\star)$, and similarly for $g_i^\star$, $\nabla g_i^\star$. Observe that the distribution rule in \eqref{eq:solutionf} and the corresponding $\rchi({\fr_\p},\bk)$ in \eqref{eq:solutionchi} are the unique solution of the linear system $g_i^\star=0$ for all $i\in[\bk-\p+1]$, that is 

\be
\label{eq:friskyformula}
\begin{cases}
&j\fr_\p(j)-\fr_\p(j+1)-\rchi({\fr_\p},\bk)=0\quad j \in [\p,\bk-1]\,,\\
&(\bk-1)\fr_\p(\bk)-\rchi({\fr_\p},\bk)=0 \,,\\
&\fr_\p(j) = f^\star_\p(j)\qquad  \hspace*{26.5mm}j\in[\p]\,.
\end{cases}
\ee
Primal feasibility and complementarity slackness are hence naturally satisfied. We are only left to prove that there exists $\mu_i\ge0$ such that $\nabla J^\star+\sum_{i=1}^{\bk-\p+1}\mu_i\nabla g_i^\star=0$. We proceed by writing the stationarity conditions explicitly and show that this is indeed the case. Note that both the cost function and the constraints are linear so that their derivatives are constant functions
\[\begin{split}
\nabla J^\star &= (1,0,\dots,0)\\
\nabla g_1^\star &= (-1, -1 ,0,\dots,0)\\
\nabla g_2^\star &= (-1,\p+1,-1,0,\dots,0)\\
\vdots&\\
\nabla g_{\bk-\p-1}^\star &=(-1,0,\dots,0,\bk-2,-1,0)\\
\nabla g_{\bk-\p} ^\star&=(-1,0,\dots,0,\bk-1,-1)\\
\nabla g_{\bk-\p+1}^\star &=(-1,0,\dots,0,\bk-1)
\end{split}
\]
Solving the stationarity condition in a recursive fashion starting from last component gives
\[\begin{cases}
\mu_{i}&=\mu_{\bk-\p+1}\frac{(\bk-1)(\bk-1)!}{(\p+i-1)!}\quad i\in[\bk-\p]\\
\sum_{i=1}^{\bk-\p+1}\mu_i&=1\,.
\end{cases}
\]
Substituting the first equation into the second one and solving yields
\[\begin{cases}
\mu_{i}=\biggl( \sum_{i=1}^{\bk-\p-1} \frac{(\bk-1)(\bk-1)!}{(\p+i-1)!}\biggr)^{-1} \hspace*{-3mm}\frac{(\bk-1)(\bk-1)!}{(\p+i-1)!}\hspace*{3mm} i\in[\bk-\p],\\
\mu_{\bk-\p+1}=\biggl( \sum_{i=1}^{\bk-\p-1} \frac{(\bk-1)(\bk-1)!}{(\p+i-1)!}\biggr)^{-1}\,.
\end{cases}
\]
Since $\mu_i\ge0$ for all $i\in[\bk-\p+1]$, we conclude that \eqref{eq:solutionf} and \eqref{eq:solutionchi} solve the relaxed program \eqref{eq:relaxed}.
\end{proof}
\begin{proof}
ii) The proof proceeds by showing that \eqref{eq:solutionf} and \eqref{eq:solutionchi} satisfy the constraints removed when transforming the original program \eqref{eq:original} into \eqref{eq:relaxed}.

Using \eqref{eq:solutionchi} and \eqref{eq:solutionf}, it is trivial to verify that $\rchi({\fr_\p},\bk)\ge0$ and $\fr_\p(1)=1$, $\fr_\p(j)\ge0$.
We proceed to prove that $\fr_\p$ is non increasing.
Note that for $j\le \p$, $\fr_\p$ coincides with $f^\star_\p$, which was proven to be non increasing in \cite{Gairing09}. Further, from Lemma \ref{lemma:generalinequality} we know that $
j\fr_\p(j)-(j+1)\fr_\p(j+1)\ge 0$ for $j \in[\p,\bk-1]$.
Thus
\[
j\fr_\p(j)-j\fr_\p(j+1)\ge j\fr_\p(j)-(j+1)\fr_\p(j+1)\ge 0
\]
for $j \in[\p,\bk-1]$, which guarantees that $\fr_\p$ is non increasing for $j\in[\p,\bk]$ too.

We are left to show that $ j\fr_\p(j)-\fr_\p(j+1)\le \rchi({\fr_\p},\bk)$  for $j \in [\p-1]$.
Since $j\in [\p-1]$, it holds that \[ j\fr_\p(j)-\fr_\p(j+1) = j f^\star_\p(j)-f^\star_\p(j+1).\] 
Note that $ j f^\star_\p(j)-f^\star_\p(j+1)\le \rchi({f^\star_\p},\p)$ by definition of $\rchi({f^\star_\p},\p)$ in \eqref{eq:chidef}. Further, $\rchi({f^\star_\p},\p)\le\rchi({f^\star_\bk},\bk)$ for any $\p\le \bk$ since the price of anarchy is a monotonically decreasing function (Proposition \ref{thm:optimalf}). Finally, Lemma \ref{lemma:worsechi} shows that for any $\p\le \bk$, it holds $\rchi({f^\star_\bk},\bk)\le\rchi({\fr_\p},\bk)$. Hence $ j\fr_\p(j)-\fr_\p(j+1)\le \rchi({\fr_\p},\bk)$ for $j\in[\p-1]$. It follows that $ \fr_\p$ is feasible for the original problem \eqref{eq:original}. 

Thanks to this, and to the fact that $\fr_\p$ is optimal for \eqref{eq:relaxed}, we conclude that $\fr_\p$ is a solution of the original problem.
\end{proof}

%%%
\section*{Proof of Lemma \ref{lemma:riskydoesnotoverpay}}
\begin{proof}
The result of Lemma \ref{lemma:generalinequality} implies that for all $\p\le\bk$
\[\p \fr_\p(\p)-j\fr_\p(j)\ge 0
\qquad \forall j\in[\p,\bk-1]\,.
\]
Further we know from \cite{Gairing09} that $\p f^\star_\p(\p)\le 1$. Since $ f^\star_\p(\p)= \fr_\p(\p)$, we conclude that for $j\in[\p,\bk-1]$
\[
1\ge\p f^\star_\p(\p)=\p\fr_\p(\p)\ge j\fr_\p(j)\,.
\]
For $j\in[\p]$, it holds $\fr_\p(j)=f^\star_\p(j)$ and we already know that the optimal distribution $f^\star_\p$ does not overpay the players \cite{Gairing09}. This concludes the proof.
\end{proof}

\section*{Proof of Theorem \ref{thm:gainsvslosses}}
\begin{proof}
i) Thanks to Proposition \ref{thm:optimalf}, the performance of $f^\star_\bk$ on the class of games with cardinality $k$ can be computed as $\poa(f^\star_\bk,k) = \frac{1}{1+\rchi(f^\star_\bk,k)}$. 
Since $k\le\bk$, we can apply part i) of Lemma \ref{lemma:redblueineq} to $\rchi(f^\star_\bk,k)$ and conclude that
\[
\poa(\fs,k) =\frac{1}{1+\rchi({f^\star_\bk},\bk)}\,.
\]
This means that the performance of $\fs$ on the set of games with cardinality $k$ is the same performance of the distribution $f^\star_\bk$ on the set of games with cardinality $\bk\ge k$, and
\[
\poa(\fs,k) =\poa({f^\star_\bk},\bk)\le\poa({f^\star_k},k)\,,
\]
where the last inequality holds since $\poa({f^\star_k},k)$ is a decreasing function of $k$ as seen in part iii) of Proposition \ref{thm:optimalf}. The inequality is tight if and only if $k=\bk$.

ii) Consider $\p\in[\bk-1]$. The performance of $\fr_\p$ on the class of games with cardinality $k$ can be computed as $\poa({\fr_\p},k) = \frac{1}{1+\rchi({\fr_\p},k)}$. Since $\p\in[k-1]$, we apply part ii) of Lemma \ref{lemma:redblueineq} to conclude that 
\[
\poa({\fr_\p},k)=\frac{1}{1+\rchi({\fr_\p},\bk)}\,.
\]
Hence, for $\p\in[\bk-1]$, the performance of $\fr_\p$ in the class of games with cardinality $k$ is the same of the performance in the class of games with cardinality $\bk$ i.e. $\poa({\fr_\p},k)=\poa({\fr_\p},\bk)$.
Finally, by Lemma \ref{lemma:worsechi} we conclude that such performance is worse than what $\fs$ can offer
\[
\poa({\fr_\p},k)<\frac{1}{1+\rchi({f_{\bk}^\star},\bk)}=\poa({\fs},k)\,.
\]

 Consider $\p\in[k,\bk]$. Since $k\le \p$, only the first $k$ entries of $\fr_\p$ will determine the performance and these are identical  to $f^\star_\p$ by definition of $\fr_\p$. Hence 
$\poa({\fr_\p},k)=\frac{1}{1+\rchi({f^\star_\p},k)}$.
Further $\p\in[k,\bk]$ and part i) of Lemma \ref{lemma:redblueineq} applies
\[
\poa({\fr_\p},k)=\frac{1}{1+\rchi({f^\star_\p},\p)}=\poa({f^\star_\p},\p)\,,
\]
so that $\fr_\p$ has the same performance of $f^\star_\p$.
Using the fact that the optimal price of anarchy is a decreasing function,  for any $p\in[k,\bk]$ we get
\[
\poa({\fr_\p},k)=\poa({f^\star_\p},\p)\ge
\poa({f^\star_\bk},\bk)=\poa({\fs},k)\,.
\]
The inequality is tight if and only if $p=\bk$.
\end{proof}
\section*{Proof of Theorem \ref{thm:alg}}
\begin{proof}
i) Consider $x_r(t)$ for fixed $\ri\in\mR$. The integer sequence $\{x_r(t)\}_{t=0}^\infty$ is upper bounded by the true cardinality $k$ (by definition of cardinality) and is non decreasing in $t$ thanks to its update rule (line 6 in Algorithm \ref{alg:learning}). Hence, after a finite number of steps, $x_r(t)$ has converged to $x_\ri^\infty$. Repeating the same reasoning for all the resources $r\in\mR$ shows that the map $x_\ri(t)$ converges in a finite number $\hat t$ of steps. Hence, for $t\ge \hat t$ the distribution rule used in the algorithm is fixed. Consequently the game is potential as it can be formulated as a standard congestion game \cite{Gairing09, MondererShapley96}.
Since for $t\ge \hat t$ agents are playing round-robin best response on a potential game, their strategy will converge in a finite number of steps to a Nash equilibrium of the game with resource specific distribution rules fixed to $\falgkinfi$ for $\ri\in\mc{R}$.

ii) Let us define $\ke=\max_{r\in\mc{R}}|\ae|_r$ (note that in general $\ke\neq\km$). To ease the notation, in the following we will simply use $f(x_r,|a|_r)$ to indicate $\falgkinfi(|a|_r)$, and similarly $u^i(\a{i},\a{-i})$ to refer to $u^i(\a{i},\a{-i};\{x_\ri^\infty\}_{r\in\mc{R}})$. Further, we define $A_e=\cup_{i}a_e^i$ and $A_o=\cup_{i}\aopt^i$

By definition of equilibrium we have for all $i\in[n]$, $u^i(\a{i}_e,\a{-i}_e)\ge u^i(\aopt^i,\a{-i}_e)$ and hence
\be
0\le \sum_{i\in[n]}u^i(\a{i}_e,\a{-i}_e)-\sum_{i\in [n]}u^i(\aopt^i,\a{-i}_e)\,.
\label{eq:eqvsopt}
\ee
Using the definition of payoff, the first term can be rewritten as
\be\begin{split}
\label{eq:firstterm}
&\sum_{i\in[n]}u^i(\a{i}_e,\a{-i}_e) = \sum_{i\in[n]}\sum_{r\in\a{i}_e} f(x_\ri,|\ae|_r)v_r
\\
=& \sum_{r\in A_e} |a_e|_r  f(x_\ri,|\ae|_r) v_r= \sum_{j=1}^{\ke}\sum_{
\substack{r\in A_e\\ |\ae|_r=j}} j f(x_\ri,j)v_r\,.
\end{split}
\ee
With a similar manipulation the second term becomes 
\[
\begin{split}
\sum_{i\in[n]}u^i(\aopt^i,\a{-i}_e) &= \sum_{i\in[n]}\sum_{r\in\aopt^i} f(x_\ri,|(\aopt^i,\a{-i}_e)|_r)v_r
\\
&\ge \sum_{i\in[n]}\sum_{r\in\aopt^i} f(x_\ri,\min\{k,|a_e|_r+1\})v_r,
\end{split}
\]
this holds because for all resources $|(\aopt^i,\a{-i}_e)|_r\le |a_e|_r+1$, $|(\aopt^i, \a{-i}_e)|_r\le k$ and $f$ is non increasing in its second argument.
For resources $r\in \aopt$ it holds $ |\aopt|_r\ge1$, and so
\[
\begin{split}
 &\sum_{i\in[n]}\sum_{r\in\aopt^i} f(x_\ri,\min\{k,|\ae|_r+1\})v_r 
 \\
  =& \sum_{r\in A_o} |\aopt|_r f(x_\ri,\min\{k,|\ae|_r+1\})v_r
  \\
  \ge& \sum_{r\in A_o} f(x_\ri,\min\{k,|\ae|_r+1\})v_r.
 \end{split}
\]
The second term in \eqref{eq:eqvsopt} can thus be lower bounded by 
\be
\label{eq:secondterm}
\begin{split}
\sum_{i\in[n]}u^i(\aopt^i,\a{-i}_e) \ge& \sum_{r\in A_o} f(x_\ri,\min\{k,|a_e|_r+1\})v_r \\
=
& \sum_{j=0}^{\ke}\sum_{
\substack{r\in A_o\\ |\ae|_r=j}} f(x_\ri,\min\{k,j+1\})v_r\,.
\end{split}
\ee
Substituting \eqref{eq:firstterm} and \eqref{eq:secondterm} in \eqref{eq:eqvsopt} gives
\allowdisplaybreaks[3]
{\small
\begin{align}
0\le& \sum_{i\in[n]}u^i(\a{i}_e,\a{-i}_e)-\sum_{i\in [n]}u^i(\aopt^i,\a{-i}_e) \nonumber \\
\le&
\sum_{j=1}^{\ke}\sum_{
\substack{r\in A_e \\ |\ae|_r=j}} j f(x_\ri,j)v_r  -
\sum_{j=0}^{\ke}\sum_{
\substack{r\in A_o\\ |\ae|_r=j}} f(x_\ri,\min\{k,j+1\})v_r \nonumber \\
=&
\sum_{j=1}^{\ke}\sum_{
\substack{r\in A_e\\ |\ae|_r=j}} j  f(x_\ri,j)v_r  -
\sum_{j=1}^{\ke}\sum_{
\substack{r\in A_o\\ |\ae|_r=j}} f(x_\ri,\min\{k,j+1\})v_r \nonumber \\
&-\sum_{r\in A_o\setminus A_e }v_r f(x_\ri,1)\nonumber \\
=&
\sum_{j=1}^{\ke}\sum_{
\substack{r\in A_e\setminus A_o\\ |\ae|_r=j}} j  f(x_\ri,j)v_r +
\sum_{j=1}^{\ke}\sum_{
\substack{r\in A_e\cap A_o\\ |\ae|_r=j}} j  f(x_\ri,j)v_r \nonumber
\\
&-
\sum_{j=1}^{\ke}\sum_{
\substack{r\in A_o\\ |\ae|_r=j}} f(x_\ri,\min\{k,j+1\})v_r
-\sum_{r\in A_o\setminus A_e }v_r f(x_\ri,1)\nonumber \\
=&
\sum_{j=1}^{\ke}\sum_{
\substack{r\in A_e\setminus A_o\\ |\ae|_r=j}} j  f(x_\ri,j)v_r 
-\sum_{r\in A_o\setminus A_e }v_r \nonumber \\
&+
\sum_{j=1}^{\ke} \sum_{
\substack{r\in A_o\\ |\ae|_r=j}}\biggl( j  f(x_\ri,j)-  f(x_\ri,\min\{k,j+1\})\biggr)v_r
\,,
\label{eq:mainineq}
\end{align}
}
where we have used the fact that $f(x_\ri,1)=1$ for all resources.
We intend to bound the first and the third term in the last expression. In the summands of \eqref{eq:mainineq} $j=|\ae|_r\le x^\infty_r$ due to the update of $x_\ri(t)$ in Algorithm \ref{alg:learning} and recall that $f(x_\ri,|a|_r)=\falgkinfi(|a|_r)$. Hence we can apply Lemma  \ref{lemma:keypoints} to the first term in \eqref{eq:mainineq} 
\be\begin{split}
\sum_{j=1}^{\ke}\sum_{
\substack{r\in A_e\setminus A_o\\ |\ae|_r=j}} j f(x_\ri,j)v_r
\le &
\sum_{j=1}^{\ke}\sum_{\substack{r\in A_e\setminus A_o\\ |\ae|_r=j}} (\rchi({f^\star_{\km},\km)}+1) v_r\,.
\end{split}
\label{eq:firstb}
\ee
Similarly for the third term in \eqref{eq:mainineq}
\be
\label{eq:secondb}
\begin{split}
&\sum_{j=1}^{\ke} \sum_{
\substack{r\in A_o\\ |\ae|_r=j}}\biggl( j f(x_\ri,j)-  f(x_\ri,\min\{k,j+1\})\biggr)v_r
\\
&\le
\sum_{j=1}^{\ke}\sum_{\substack{r\in A_o\\ |\ae|_r=j}} \rchi({f^\star_{\km},\km)} v_r=\rchi({f^\star_{\km},\km)}  \hspace*{-2mm}\sum_{r\in \ae\cap \aopt} \hspace*{-2mm} v_r\,.
\end{split}
\ee
Hence combining \eqref{eq:mainineq} with the bounds from \eqref{eq:firstb} and \eqref{eq:secondb}
\begin{align}
0\le &(\rchi(f^\star_{\km},\km)+1) \sum_{j=1}^{\ke}\sum_{
\substack{r\in A_e\setminus A_o\\ |\ae|_r=j}}v_r
\\
&+\rchi(f^\star_{\km},\km)  \sum_{r\in A_e\cap A_o} v_r 
-
\sum_{r\in A_o\setminus A_e }v_r \nonumber\\
=&(\rchi(f^\star_{\km},\km)+1)\sum_{r\in A_e\setminus A_o}v_r 
\\
&+
\rchi(f^\star_{\km},\km)  \sum_{r\in A_e\cap A_o} v_r -\sum_{r\in A_o\setminus A_e }v_r \nonumber\\
=&(\rchi(f^\star_{\km},\km)+1)\sum_{r\in A_e\setminus A_o}v_r 
\\
&+
(\rchi(f^\star_{\km},\km)+1)  \sum_{r\in A_e\cap A_o}  v_r 
-
\sum_{r\in A_o}v_r \nonumber\\
=&(\rchi(f^\star_{\km},\km)+1)\sum_{r\in A_e} v_r
-
\sum_{r\in A_o}v_r \nonumber\\ 
=&(\rchi(f^\star_{\km},\km)+1)W(\ae)-W(\aopt).\nonumber
\end{align}
Hence $(\rchi(f^\star_{\km},\km)+1)W(\ae)-W(\aopt)\ge 0$ and rearranging
\[
\frac{W(\ae)}{W(\aopt)}\ge \frac{1}{1+\rchi(f^\star_{\km},\km)}=\poa(f^\star_{\km},{\km})\ge\poa({f^\star_{k}},{k})\,,
\]
where the last inequality follows from the fact that the price of anarchy is a decreasing function, and $\km\le k$ by definition of cardinality.
\end{proof}

\section*{Proof of Theorem \ref{thm:counterex}}
\begin{proof}
i)
Consider the covering problem depicted in the following figure (a), composed of players $p_1$, $p_2$, $p_3$ represented by a solid dot; resources $r_1$, $r_2$, $r_3$, $r_4$ represented by a circle with values $v_1$, $v_2$, $v_3$, $v_4$ such that 
\[v_1>v_3>v_4>v_2 \quad \text{and}\quad v_1f^\star_3(2)<v_2<v_1f^\star_2(2)<v_4\] 
As an example take $v=(11,5,7,6)$. 
Each player $p_1$, $p_2$, $p_3$ can choose only one resource from $\{r_1,r_2,r_3\}$, $\{r_2,r_3,r_4\}$, $\{r_1,r_2,r_3,r_4\}$, respectively i.e. each player can only choose one arrow pointing outwards from himself.

\begin{figure}[h!]
          \begin{subfigure}[b]{.5\linewidth}
            \centering \includegraphics[scale = 0.65]{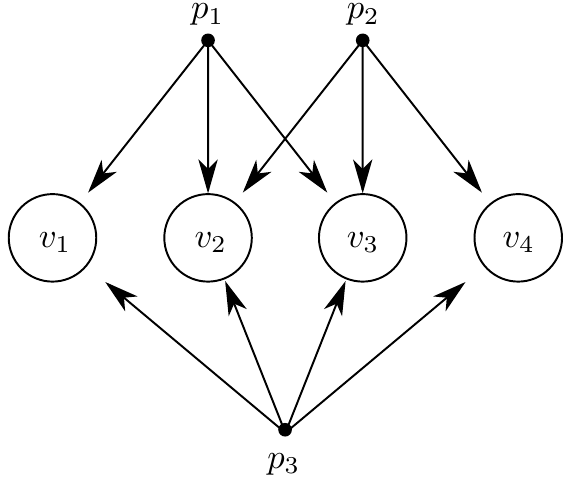}
            \caption{Original game}\label{fig:originalmine}
          \end{subfigure}%
          \begin{subfigure}[b]{.5\linewidth}
            \centering \includegraphics[scale = 0.65]{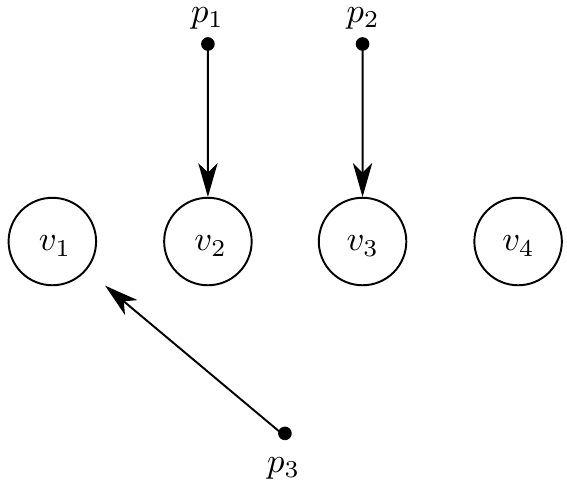}
            \caption{Equilibrium $\ae$}\label{fig:equilmine}
          \end{subfigure}
\end{figure}

The cardinality is $k=3$ since all players could choose simultaneously $r_2$ or $r_3$, hence the optimal distribution rule is $f^\star_3$. Amongst the equilibria obtained with $f^\star_3$ there is $a_e=(r_2,r_3,r_1)$, depicted in the previous figure (b). This configuration is an equilibrium since 
$v_2> v_1 f^\star_3(2) > v_3 f^\star_3(2)$ and $v_3> v_4$, $v_3> v_2 f^\star_3(2)$ and $v_1> v_2 f^\star_3(2)$, $v_1> v_3 f^\star_3(2)$, $v_1> v_4$. Such equilibrium gives a welfare of $v_1+v_2+v_3$ that is less than the optimal $v_1+v_3+v_4$, since $v_2<v_4$. We intend to show that for any initial condition and for any execution, Algorithm \ref{alg:learning} will converge to an optimal allocation. This suffices to prove that the worst equilibrium obtained with Algorithm \ref{alg:learning} performs better than the worst equilibrium obtained with $f^\star_3$, which is not optimal as shown above.
Observe that the conditions $v_1>v_3>v_4>v_2$ and $v_4>v_1f^\star_2(2)$ ensure that an allocation with two or more agents covering the same resource is never an equilibrium. This holds regardless of the distribution used.
Hence, the welfare can potentially take 
$
\begin{pmatrix}
	4\\ 3 
\end{pmatrix}
=
4
$ different values, since the binomial represents the number of subsets with $3$ elements (agents allocations) that can be extracted from a set of $4$ elements (set of resources).
These different welfare values are obtained for  $(r_1,r_2,r_4)$, $(r_2,r_3,r_4)$, $(r_1,r_3,r_4)$, $(r_1,r_2,r_3)$, or feasible permutations of each.
The allocation $(r_1,r_2,r_4)$ is never an equilibrium since player $p_3$ can improve moving to $r_3$ because  $v_3>v_4$. Similarly for any feasible permutation of $(r_1,r_2,r_4)$, the player selecting resource $r_4$ can always improve moving to $r_3$. 
The allocation $(r_2,r_3,r_4)$ is never an equilibrium since player
$p_3$ can improve moving to $r_1$ since $v_1$ is the highest. Similarly for any feasible permutation of $(r_2,r_3,r_4)$, there exists a player that can improve moving to $r_1$. This holds regardless of what distribution rule is used.
The allocation $(r_1,r_3,r_4)$ (or any feasible permutation) is optimal.
We are thus left to show that Algorithm \ref{alg:learning} never converges to $(r_1,r_2,r_3)$, or any other feasible permutation. We show this by enumeration.

The allocation $(r_1,r_2,r_3)$ can not be an equilibrium since player $p_2$ can improve moving to $r_4$ because $v_4>v_2$. 
The allocation $(r_1,r_3,r_2)$ can not be an equilibrium since player $p_3$ can improve moving to $r_4$. 
The allocation $(r_3,r_2,r_1)$ can not be an equilibrium since player $p_2$ can improve moving to $r_4$. 
We are left to check $a_e=(r_2,r_3,r_1)$, depicted in the previous figure (b). This can not be an equilibrium of Algorithm \ref{alg:learning}, because $v_2 < v_1 \falg(2)$ for $l=1$, $2$ and so player $p_2$ could improve moving to $r_1$. The fact that the algorithm uses $l\le 2$ on resource $r_1$ holds because the maximum number of players on $r_1$ is two, and so $k_t(1)\le 2$ at any time step $t\in [n]$. We conclude that all the equilibria towards which the algorithm converges give optimal welfare, while
$f^\star_3$ also produces the suboptimal equilibrium $a_e$; the claim follows. Observe that this is not a worst case instance because the price of anarchy with the example values $v=(11,5,7,6)$ is 
\[
\frac{W(\ae)}{W(\aopt)}=\frac{11+5+7}{11+6+7}=\frac{23}{24} >\poa({f^\star_3},3)=\frac{7}{11}\,.
\]

ii) Consider the covering problem depicted in the following figure (a), composed of players $p_1$, $p_2$, $p_3$ represented by a solid dot; resources $r_1$, $r_2$, $r_3$ represented by an empty circle with values $v_1$, $v_2$, $v_3$ such that \[v_3f^\star_3(2)<v_1<v_2<v_3/2<v_3.\] As an example take $v= [9, 9.5,20]$. 
Each player $p_1$, $p_2$, $p_3$ can choose only one resource from $\{r_1,r_2\}$, $\{r_2,r_3\}$, $\{r_1,r_2,r_3\}$, respectively i.e. each player can only choose one arrow pointing outwards from himself.

\begin{figure}[h!]
          \begin{subfigure}[b]{.5\linewidth}
            \centering \includegraphics[scale = 0.19]{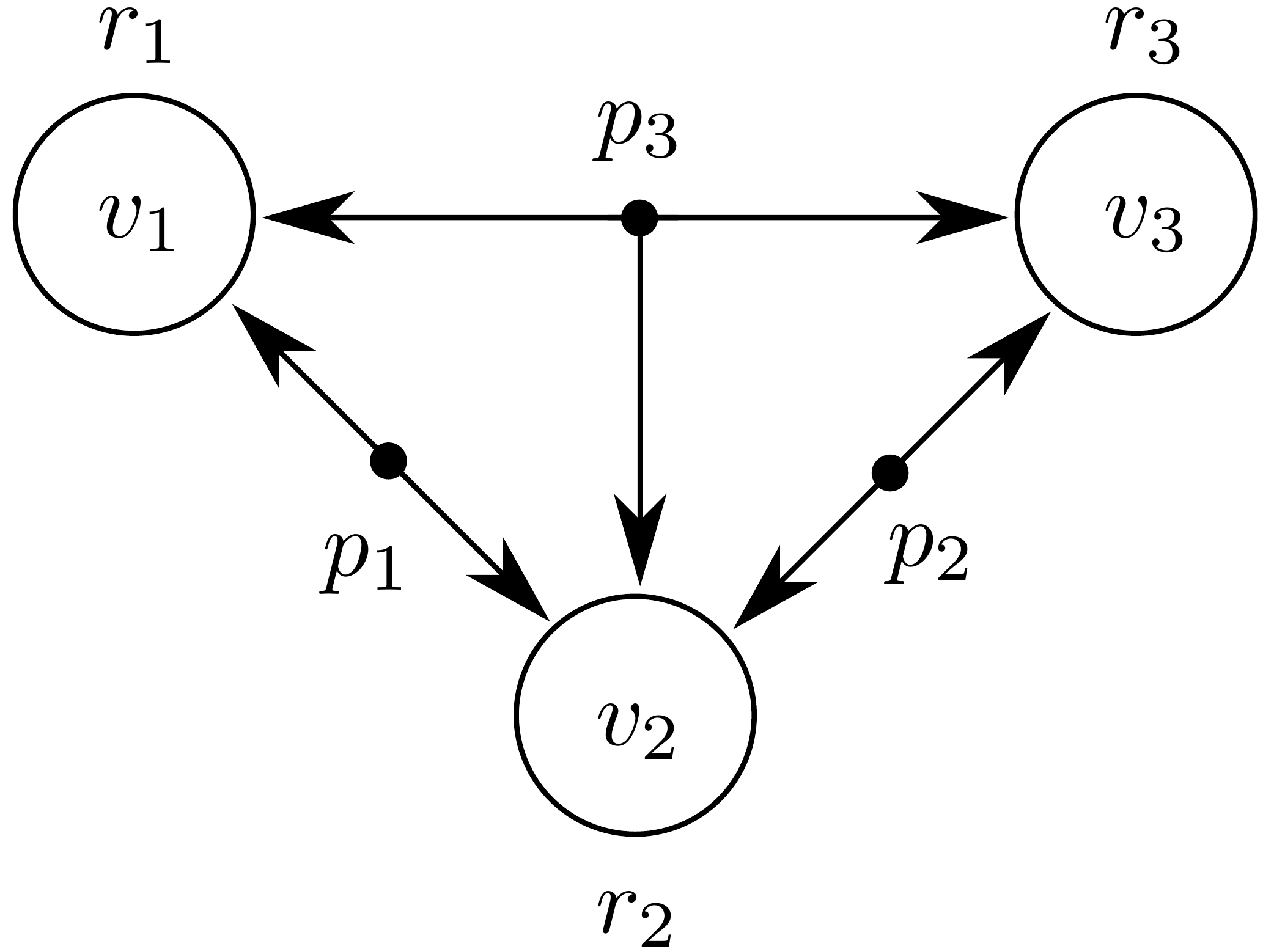}
            \caption{Original game}\label{fig:original}
          \end{subfigure}%
          \begin{subfigure}[b]{.5\linewidth}
            \centering \includegraphics[scale = 0.19]{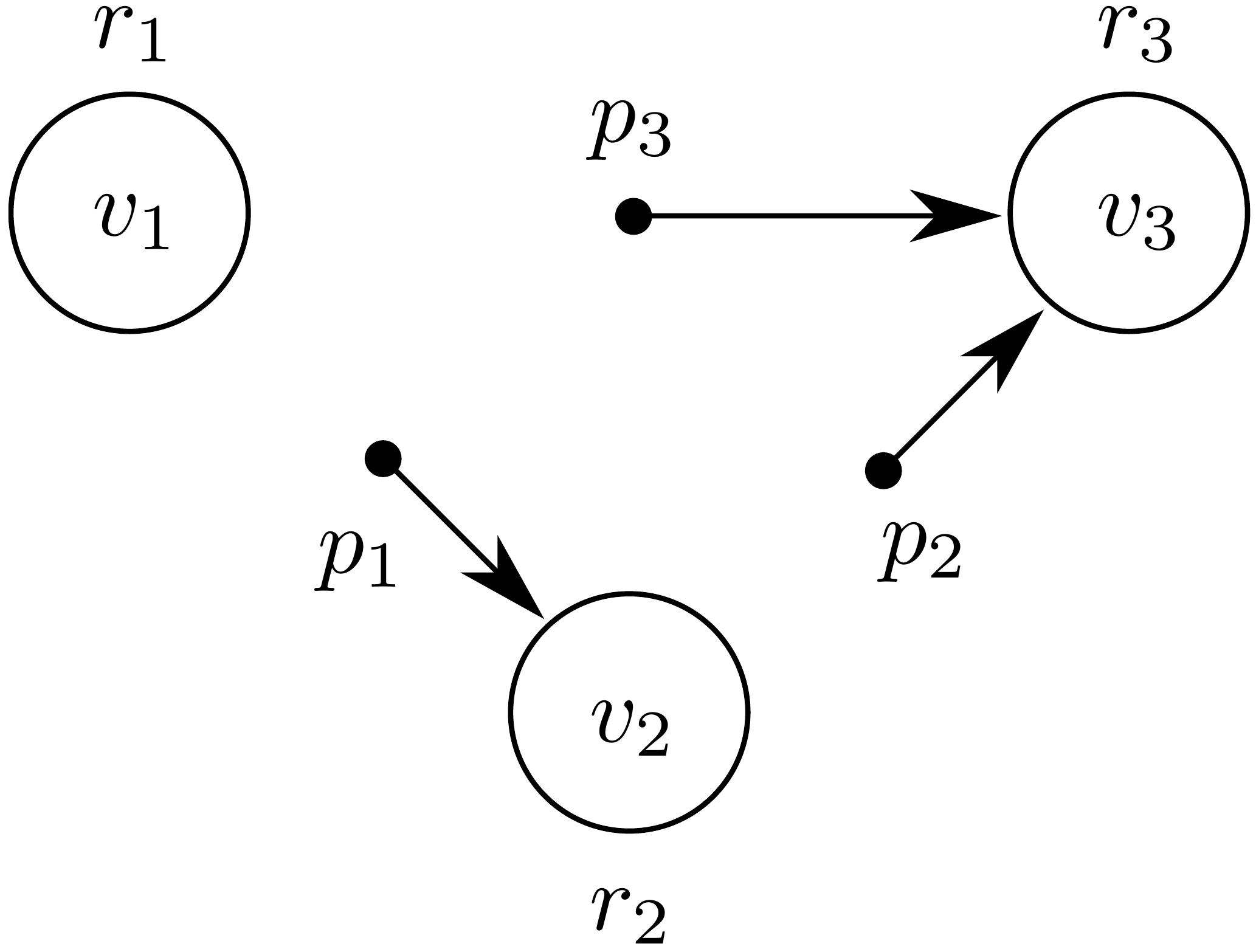}
            \caption{Equilibrium $a_2$}\label{fig:equil}
          \end{subfigure}
\end{figure}

The cardinality is $k=3$ since all players could choose simultaneously $r_1$, hence the optimal distribution rule is $f^\star_3$. All the equilibria obtained with $f^\star_3$ are completely spread i.e. they feature one and only one player on each resource. Any allocation where there are two or more players in one resource is not an equilibrium for $f^\star_3$, as detailed in the following.\\
 If all three players selected resource $r_2$, $p_2$ could improve moving to $r_3$ since $v_3>v_2$. If $p_1$ and $p_3$ selected $r_1$, depending on the choice of $p_2$, either $p_1$ or $p_3$ could improve moving respectively to $r_2$ or $r_3$ since $v_2>v_1$ and $v_3>v_1$. If $p_2$ and $p_3$ selected $r_3$, depending on the choice of $p_1$, either $p_2$ or $p_3$ could improve moving respectively to $r_2$ or $r_1$ since $v_3 f^\star_3(2) <v_2$ and $v_3 f^\star_3(2)  <v_1$. If $p_1$, $p_3$ selected both $r_2$, regardless of the choice of $p_2$, $p_3$ could improve moving to $r_3$ since $v_3>v_2$.  If $p_2$, $p_3$ selected both $r_2$, regardless of the choice of $p_1$, $p_3$ could improve moving to $r_3$ since $v_3>v_2$. Finally, if $p_1$, $p_2$ selected $r_2$, regardless of the choice of $p_3$, $p_2$ could always improve moving to $r_3$ since $v_3>v_2$.
 Thus all equilibria obtained with $f^\star_3$ (including the worst) give a welfare of $v_1+v_2+v_3$.

Let us consider Algorithm \ref{alg:learning} and initialise it at $a_1=[r_2,r_3,r_1]$, giving $k_1(r)=1$ for all $r$. Player $p_3$ updates and since $v_3\cdot 1>v_1\cdot 1$, he selects $r_3$, giving $a_2=[r_2,r_3,r_3]$ and $k_2(r)=1$ for $r_1$, $r_2$ and $k_2(r_3)=2$. This allocation is depicted in the previous figure (b) and is an equilibrium configuration. Indeed $p_1$ can not improve since $v_2>v_1$; $p_2$ can not improve since $v_3 f^\star_2(2)= \frac{v_3}{2}>v_2 \cdot 1 $; $p_3$ can not improve since $v_3 f^\star_2(2)= \frac{v_3}{2}>v_2 \cdot 1 $ and $v_3 f^\star_2(2)= \frac{v_3}{2}>v_1 \cdot 1 $. Such equilibrium has a welfare of $v_2+v_3$.

In conclusion, all equilibria obtained with $f^\star_3$ give a better welfare than $a_2$ and thus of the worst equilibrium obtained with Algorithm \ref{alg:learning}.

\end{proof}

%%%%LEMMAS
\section{Lemmata}
\label{app:lemmata}
\begin{lemma}
\label{lemma:generalinequality}
Let $\p\in[\bk]$. The distribution $\fr_\p$ satisfies 
\[j  \fr_\p(j)-(j+1)\fr_\p(j+1)\ge 0
\qquad j \in[\p,\bk-1]\,.
\]
\begin{proof}
Recall that $\fr_\p$ is obtained from equation \eqref{eq:friskyformula}. Using $\rchi({\fr_\p},\bk)$ from \eqref{eq:solutionchi}, one can reconstruct the tail entries of $\fr_\p(j)$ with the following backward recursion
\[
\begin{split}
&j \fr_\p(j)-\fr_\p(j+1)=\rchi({\fr_\p},\bk)\quad j \in [\p,\bk-1]\,,\\
&(\bk-1)\fr_\p(\bk)=\rchi({\fr_\p},\bk)\,.
\end{split}
\]
Starting from $\fr_\p(\bk)=\frac{\rchi({\fr_\p},\bk)}{\bk-1}$, the first equation gives for $j\ge \p$
\[
j  \fr_\p(j) =\rchi({\fr_\p},\bk)\biggl(1+\sum_{i=j+1}^{\bk-1} \frac{j!}{i!}+\frac{j!}{(\bk-1)(\bk-1)!}\biggl)\,.
\]
Hence
\[\small \begin{split}
&\frac{1}{\rchi({\fr_\p},\bk)}(j  \fr_\p(j) -(j+1)\fr_\p(j+1)) =\\
&=\sum_{i=j+1}^{\bk-1} \frac{j!}{i!}-\sum_{i=j+2}^{\bk-1} \frac{(j+1)!}{i!}+\frac{j!-(j+1)!}{(\bk-1)(\bk-1)!}=\\
&=\sum_{i=j+1}^{\bk-2} \biggl( \frac{j!}{i!}-\frac{(j+1)!}{(i+1)!}\biggr)+\frac{j!}{(\bk-1)!}\biggl(1+\frac{1}{\bk-1}-\frac{j+1}{\bk-1}\biggr)\,.
\end{split}
\]
Note that for $i>j$, one has 
\[\frac{j!}{i!}=\frac{1}{i(i-1)\dots(j+1)}\text{ ~ and so ~} \frac{j!}{i!}-\frac{(j+1)!}{(i+1)!}>0.\]
Further 
\[
1+\frac{1}{\bk-1}+\frac{j+1}{\bk-1} = \frac{\bk -j-1}{\bk-1}\ge0
\]
since $j\le \bk-1$ by assumption.
Hence we conclude that
\[j  \fr_\p(j)-(j+1)\fr_\p(j+1)\ge 0
\qquad j \in[\p,\bk-1]\,.
\]
\end{proof}
\end{lemma}
\begin{lemma}
\label{lemma:worsechi}
For any $1<\p<\bk$ it holds $\rchi({f^\star_\bk},\bk)<\rchi({\fr_\p},\bk)$. 
\end{lemma}
\begin{proof}
The expression of $\rchi({f^\star_\bk},\bk)$ in \eqref{eq:optimalchi} and of $\rchi({\fr_\p},\bk)$ in equation \eqref{eq:solutionchi} can be rewritten as 
\[\begin{split}
\rchi({f^\star_\bk},\bk) &=
\frac{(\bk-1)(\bk-1)!}{1+(\bk-1)(\bk-1)!\sum_{i=1}^{\bk-1}\frac{1}{i!}}\,,\\
\rchi({\fr_\p},\bk) &= \frac{(\bk-1)(\bk-1)!}{\bk+\sum_{h=1}^{\bk-1-\p}\frac{(\bk-1)(\bk-1)!}{(\bk-h-1)!}}\beta(\p)\,,
\end{split}
\]
where 
\[
\beta(\p)\coloneqq
\frac{1}{1+\sum_{h=1}^{\p-1}\frac{(\p-1)(\p-1)!}{h!}}\,.
\]
Instead of showing $\rchi({f^\star_\bk},\bk)<\rchi({\fr_\p},\bk) $, in the following we equivalently prove that $\frac{1}{\rchi({f^\star_\bk},\bk)}>\frac{1}{\rchi({\fr_\p},\bk)}$ i.e., that
{\small{ \[\begin{split}
&{1+(\bk-1)(\bk-1)!\sum_{i=1}^{\bk-1}\frac{1}{i!}}=
{\bk+(\bk-1)(\bk-1)!\sum_{i=1}^{\bk-2}\frac{1}{i!}} >\\
&\biggl({\bk+\!\!\!\sum_{h=1}^{\bk-1-\p}\frac{(\bk-1)(\bk-1)!}{(\bk-h-1)!}}\biggr)\biggl({1+\!\sum_{h=1}^{\p-1}\frac{(\p-1)(\p-1)!}{h!}}\biggr).
\end{split}
\]}}
The previous inequality can be rewritten as
{\small {\[
\begin{split}
&
{(\bk-1)(\bk-1)!\biggl(\sum_{i=1}^{\bk-2}\frac{1}{i!}}-\sum_{h=1}^{\bk-1-\p}\frac{1}{(\bk-h-1)!} \biggr)>\\
&\biggl({\bk+\sum_{h=1}^{\bk-1-\p}\frac{(\bk-1)(\bk-1)!}{(\bk-h-1)!}}\biggr)\biggl({\sum_{h=1}^{\p-1}\frac{(\p-1)(\p-1)!}{h!}}\biggr)\,.
\end{split}
\]}}
Since the left hand side is equal to $(\bk-1)(\bk-1)!\sum_{h=1}^{\p-1}\frac{1}{h!}$, we can simplify the term $\sum_{h=1}^{\p-1}\frac{1}{h!}$ to get
\[
\small
\begin{split}
&
\frac{(\bk-1)(\bk-1)!}{(\p-1)(\p-1)!}>\biggl({\bk+\!\!\sum_{h=1}^{\bk-1-\p}\frac{(\bk-1)(\bk-1)!}{(\bk-h-1)!}}\biggr)\,,
\end{split}
\]
which is finally equivalent to 
\be
\label{eq:finalineq}
\frac{1}{(\p-1)(\p-1)!}>\frac{\bk}{(\bk-1)(\bk-1)}+\!\!\!\sum_{h=1}^{\bk-1-\p}\!\!\!\!\frac{1}{(\bk-h-1)!}.
\ee
We use \emph{induction} to show that inequality \eqref{eq:finalineq} holds for $1<\p<\bk$, as required. We start from $\p=\bk-1$ and apply induction backwards until we reach $\p=2$.

i) For $\p=\bk-1$ and $\p>1$ inequality \eqref{eq:finalineq} reads as 
\[
\frac{1}{(\p-1)(\p-1)!}>\frac{\p+1}{\p\cdot \p!}\iff \p^2>\p^2-1\,,
\]
which is always satisfied.

ii) Let us assume the inequality holds for a generic $\p\le\bk-1$, we show that it holds also for $\p-1$ (with $\p>1$). That is, we assume
\be
\label{eq:assind}
\frac{1}{(\p-1)(\p-1)!}>\frac{\bk}{(\bk-1)(\bk-1)}+\!\!\!\sum_{h=1}^{\bk-1-\p}\!\!\!\!\!\frac{1}{(\bk-h-1)!}\,,
\ee
and want to show 
\be
\label{eq:wanttoshow}
\frac{1}{(\p-2)(\p-2)!}>\frac{\bk}{(\bk-1)(\bk-1)}+\!\!\!\sum_{h=1}^{\bk-\p}\!\!\!\frac{1}{(\bk-h-1)!}\,.
\ee
We can rewrite the right hand side of \eqref{eq:wanttoshow} and use \eqref{eq:assind} to upper bound it
\be
\label{eq:chainineq}
\begin{split}
&\frac{\bk}{(\bk-1)(\bk-1)}+\sum_{h=1}^{\bk-\p}\frac{1}{(\bk-h-1)!}=\\
&
\frac{\bk}{(\bk-1)(\bk-1)}+\sum_{h=1}^{\bk-\p-1}\frac{1}{(\bk-h-1)!}
+\frac{1}{(\p-1)!}<\\
&\frac{1}{(\p-1)(\p-1)!}+\frac{1}{(\p-1)!}=
\frac{\p}{(\p-1)(\p-1)!}<\\
&\frac{1}{(\p-2)(\p-2)!}\,.
\end{split}
\ee
The last inequality holds since it is equivalent to 
\[
\frac{\p}{(\p-1)^2}<\frac{1}{\p-2}\iff\p^2-2\p<(\p-1)^2\,,
\]
which is always satisfied. Comparing the first and last term in \eqref{eq:chainineq} gives \eqref{eq:wanttoshow}.

This completes the induction and thus the proof.

\end{proof}
%%%
\begin{lemma}
\label{lemma:redblueineq}
\begin{itemize}
\item[i)] For any $1\le l \le m$,~$m\in\mb{N}$ it holds 
\[\rchi({f^\star_{m}},l)=\rchi({f^\star_{m}},m)\,.\]
\item[ii)] For any $k\in[\bk]$ and $1<\p<k$ it holds 
\[\rchi({\fr_\p},k)=\rchi({\fr_\p},\bk)\,.\]
\end{itemize}
\end{lemma}
\begin{proof}
i) If $l=m$, the result holds trivially. Hence in the following we consider $l\in[m-1]$. By definition of $\rchi({f^\star_{m}},l)$ in \eqref{eq:chidef}, one has
 \[
\begin{split}
\rchi({f^\star_{m}},l)=&\min_{x\ge0} x\\
& \text{s.t.}\quad j f^\star_{m}(j)-f^\star_{m}(j+1)\le x\quad j\in[l-1]\,,\\
& \qquad\, (l-1)f^\star_{m}(l)\le x\,.
\end{split}
 \]
Note that $f^\star_{m}$ is derived in \cite[Theorem 2]{Gairing09} solving the following recursion
 \be
 \label{eq:recur}
 \begin{split}
& j f^\star_{m}(j)-f^\star_{m}(j+1)= \rchi({f^\star_{m}},m) \quad   j\in[m-1]\\
&(l-1)f^\star_{m}(l)=\rchi({f^\star_{m}},m)\,.
\end{split}\ee
Since $m>l$, it follows that any feasible $x$ from the LP above has to satisfy $x\ge \rchi({f^\star_{m}},m)$. In the following we show that setting $x=\rchi({f^\star_{m}},m)$, the constraint $(l-1)f^\star_{m}(l)\le x$ is also satisfied. This will be enough to conclude that $\rchi({f^\star_{m}},l)=\rchi({f^\star_{m}},m)$.\\
Since $f^\star_{m}$ is non increasing, one has
\[\begin{split}
&(l-1)f^\star_{m}(l)-\rchi({f^\star_{m}},m) =lf^\star_{m}(l)-f^\star_{m}(l)-\rchi({f^\star_{m}},m)\le\\
&\le lf^\star_{m}(l)-f^\star_{m}(l+1)-\rchi({f^\star_{m}},m) =0 \,,
\end{split}
\]
where the equality holds applying  \eqref{eq:recur} for $j=l\in[m-1]$.

ii) We intend to compute 
 \[
 \begin{split}
\rchi({\fr_\p},k)= &\min_{x\ge0} x\\
& \text{s.t.}\quad j \fr_\p(j)-\fr_\p(j+1)\le x\quad j \in [k-1]\\
& (k-1)\fr_\p(k)\le x\,.
\end{split}
 \]
For any feasible $x$, it must be $x\ge \rchi({\fr_\p},\bk)$ due to how $ \fr_\p(j)$ is recursively defined for $j>\p$ in Equation \eqref{eq:friskyformula}. Similarly to what shown before, one can prove that $x = \rchi({\fr_\p},\bk)$ will also satisfy the constraint $(k-1)\fr_\p(k)\le x$. Hence $\rchi({\fr_\p},k)=\rchi({\fr_\p},\bk)$ and the proof is concluded.
\end{proof}
%%%
\begin{lemma}
\label{lemma:keypoints}
For all resources $r\in\mR$, the distribution rules $\falgkinfi$ are such that
\begin{align}
&j\,\falgkinfi(j)\le \rchi({f^\star_{\km}},\km)+1 &&j\in[x^\infty_\ri]
\label{eq:lemmapoint1}
\\
&j\,\falgkinfi(j)-\falgkinfi(\min\{k,j+1\})\le \rchi({f^\star_{\km}},\km)&& j\in[x^\infty_\ri]
\label{eq:lemmapoint2}
\end{align}
where $\km=\max_{r\in\mR}x^\infty_\ri$.
\end{lemma}
\begin{proof}
We start from \eqref{eq:lemmapoint1} and examine $\falgkinfi$ for a fixed $r\in\mR$. Consider $j\in[x^\infty_\ri-1]$, by definition of $\falgkinfi$ and the fact that $f^\star_{x^\infty_\ri}$ is non increasing
\[\begin{split}
&j\,\falgkinfi(j)-1=j\,f^\star_{x^\infty_\ri}(j)-f^\star_{x^\infty_\ri}(1)\\
&\le j\,f^\star_{x^\infty_\ri}(j)-f^\star_{x^\infty_\ri}(j+1)
\le
\rchi({f^\star_{x^\infty_\ri}},{x^\infty_\ri}) \\
&\implies\quad j\,\falgkinfi(j)\le \rchi({f^\star_{x^\infty_\ri}},{x^\infty_\ri}) + 1\,,
\end{split}
\]
where the last inequality holds thanks to the definition \eqref{eq:chidef}.
Since $x^\infty_\ri \le \km$ and the price of anarchy is a decreasing function, one has $\rchi(f^\star_{x^\infty_\ri},{x^\infty_\ri}) \le \rchi({f^\star_{\km}},\km)$ and so $j\,\falgkinfi(j)\le \rchi({f^\star_{\km}},\km)+1$ for $j\in[x^\infty_\ri-1]$.\\
In a similar fashion when $j=x^\infty_\ri$
\be
\label{eq:previous1}
\begin{split}
&x^\infty_\ri\falgkinfi(x^\infty_\ri)-1= x^\infty_\ri f^\star_{x^\infty_\ri}(x^\infty_\ri)-f^\star_{x^\infty_\ri }(1)\\
&\le x^\infty_\ri f^\star_{x^\infty_\ri }(x^\infty_\ri)-f^\star_{x^\infty_\ri }(x^\infty_\ri)\\
&=
(x^\infty_\ri-1)f^\star_{x^\infty_\ri }({x^\infty_\ri})\\
&=\rchi({f^\star_{x^\infty_\ri}},x^\infty_\ri) \le \rchi({f^\star_{\km}},\km)\,,
\end{split}
\ee
where the only difference is in the last equality that comes from equation \eqref{eq:optimalchi}.
Repeating the same reasoning for all $r\in\mR$, one has proven
\eqref{eq:lemmapoint1}.

In the remaining, we show that \eqref{eq:lemmapoint2} holds. Consider $\falgkinfi$ for a fixed resource $r\in\mR$ and recall that 
$x^\infty_\ri\le k$. Thus for $j\in[x^\infty_\ri-1]$ one has $\min\{k,j+1\}=j+1$ and the claim reads as $j\,\falgkinfi(j)-\falgkinfi(j+1)\le \rchi({f^\star_{\km}},\km)$. This holds since
\[
\begin{split}
j\,\falgkinfi(j)-\falgkinfi(j+1)&= 
j\,f^\star_{x^\infty_\ri}(j)-f^\star_{x^\infty_\ri}(j+1)\\
&
\le
\rchi({f^\star_{x^\infty_\ri}},x^\infty_\ri) \le \rchi({f^\star_{\km}},\km)
\end{split}
\]
where the first inequality holds thanks to definition \eqref{eq:chidef} and the last since the price of anarchy is non increasing ($x^\infty_\ri\le \km$). \\
In the remaining we focus on $j=x^\infty_\ri$ and divide the proof in two subparts.
When $k=x^\infty_\ri$, $\min\{k,j+1\}=x^\infty_\ri$ and the claim follows from
\[\begin{split}
&x^\infty_\ri\falgkinfi(x^\infty_\ri)-\falgkinfi(x^\infty_\ri)
=
x^\infty_\ri f^\star_{x^\infty_\ri}({x^\infty_\ri})-f^\star_{x^\infty_\ri}({x^\infty_\ri})\\
&=
(x^\infty_\ri-1)f^\star_{x^\infty_\ri}({x^\infty_\ri})
=
\rchi({f^\star_{x^\infty_\ri}},x^\infty_\ri) \le \rchi({f^\star_{\km}},\km)\,,
\end{split}
\]
similarly to \eqref{eq:previous1}.
When $k>x^\infty_\ri$, then $\min\{k,j+1\}=x^\infty_\ri+1$ and the claim holds if we show 
\[x^\infty_\ri\falgkinfi(x^\infty_\ri)-\falgkinfi(x^\infty_\ri+1)\le
\rchi({f^\star_{x^\infty_\ri}},x^\infty_\ri).\] 
For this to hold, one has to require 
\[\begin{split}
\falgkinfi(x^\infty_\ri+1)&\ge x^\infty_\ri\falgkinfi(x^\infty_\ri)-\rchi({f^\star_{x^\infty_\ri}},x^\infty_\ri)\\
&=x^\infty_\ri f^\star_{x^\infty_\ri}(x^\infty_\ri)-\rchi({f^\star_{x^\infty_\ri}},x^\infty_\ri)\\
&=f^\star_{x^\infty_\ri}(x^\infty_\ri)=\falgkinfi(x^\infty_\ri)\,,
\end{split}
\]
where the second equality sign follows form $\rchi({f^\star_{x^\infty_\ri}},x^\infty_\ri) = (x^\infty_\ri-1)f^\star_{x^\infty_\ri}(x^\infty_\ri)$, that is form equation \eqref{eq:chidef}. Hence we need to impose 
\[
\falgkinfi(x^\infty_\ri+1)\ge\falgkinfi(x^\infty_\ri)\,,
\]
but at the same time we are limited to non increasing distribution rules. Hence we set $\falgkinfi(x^\infty_\ri+1)=\falgkinfi(x^\infty_\ri)$ as by definition of $\falg$ from Equation \eqref{eq:falg}. The proof is completed by observing that the same reasoning can be repeated for any resource $r\in\mR$.
\end{proof}
%%
% *** BIBLIO ***
\bibliographystyle{IEEEtran}
\bibliography{biblio_covering_games}

\begin{IEEEbiography}
[{\includegraphics[width=1in,height=1.25in,clip,keepaspectratio]{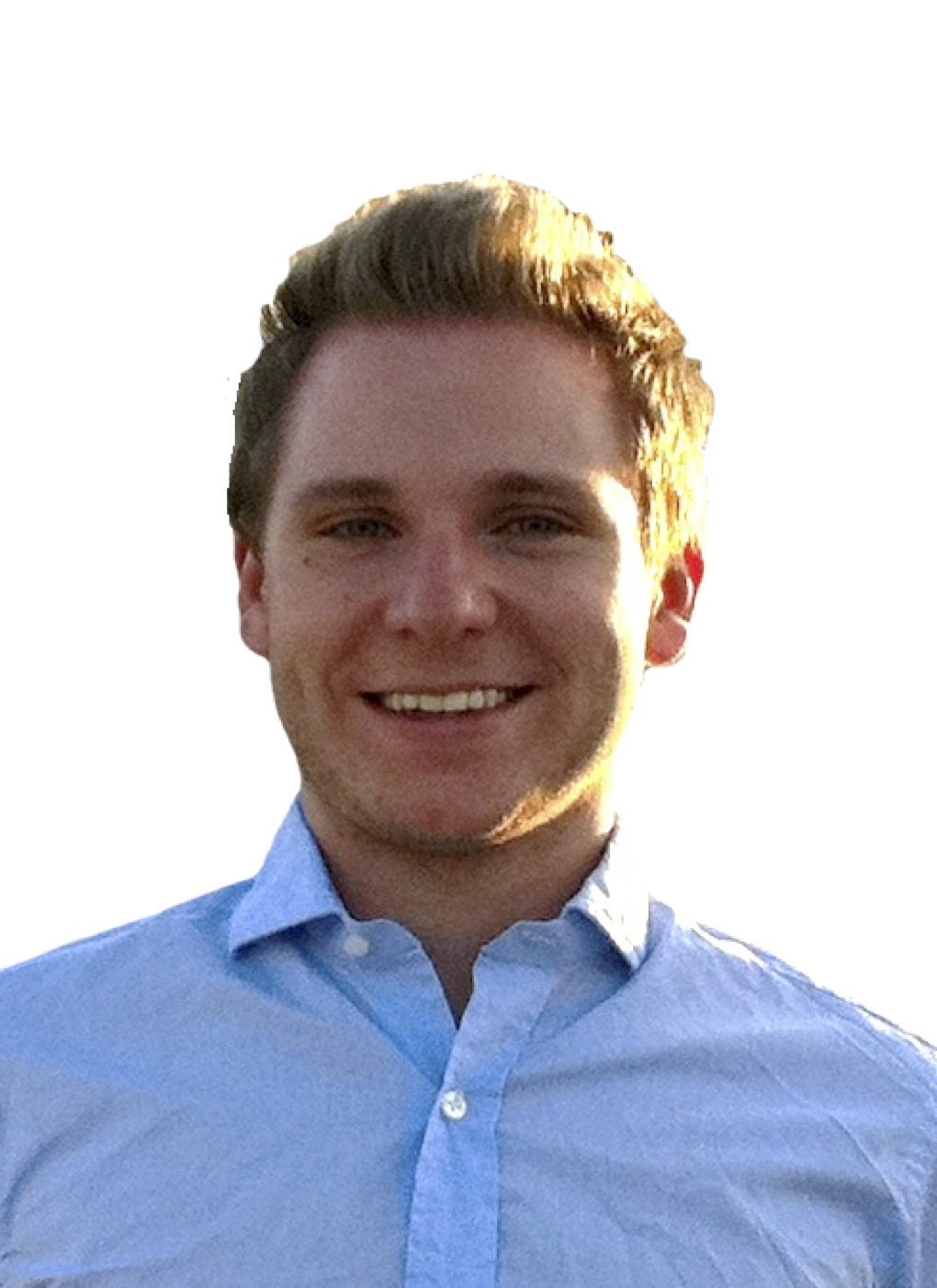}}]
{Dario Paccagnan} is a doctoral student at the Automatic Control Laboratory, ETH Zurich, Switzerland, since October 2014. He received his B.Sc. and M.Sc. in Aerospace Engineering from the University of Padova, Italy, in 2011 and 2014. In the same year he received the M.Sc. in Mathematical Modelling from the Technical University of Denmark, all with Honors. His Master's Thesis was prepared when visiting Imperial College of London, UK, in 2014. From March to August 2017 he has been a visiting scholar at the University of California, Santa Barbara. Dario is recipient of the SNSF fellowship for his work in Distributed Optimization and Game Design. His research interests are at the interface between distributed control and game theory. Applications include multiagent systems, smart cities and traffic networks.
\end{IEEEbiography}

\begin{IEEEbiography}
[{\includegraphics[width=1in,height=1.25in,clip,keepaspectratio]{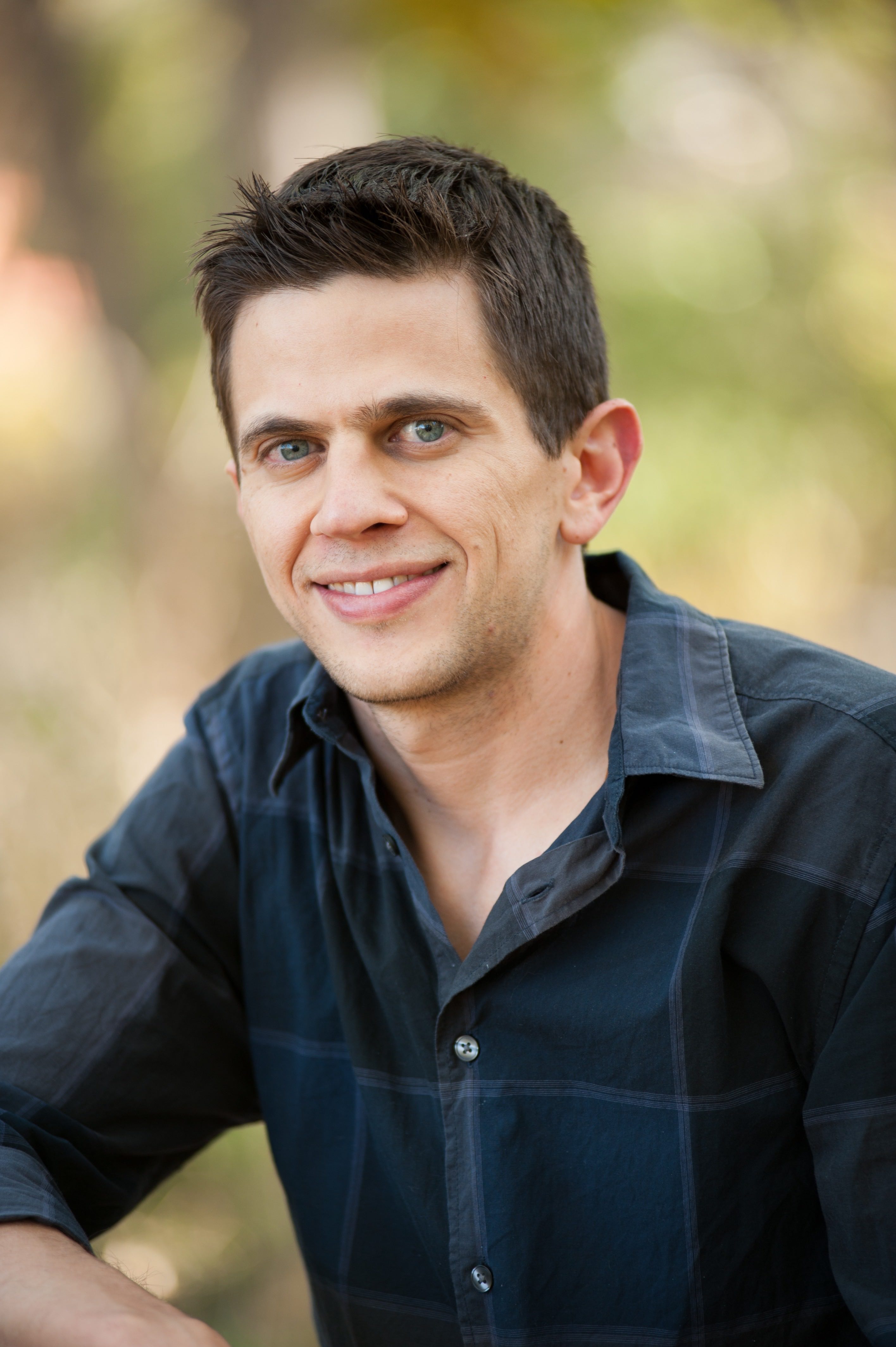}}]
{Jason Marden} is an Associate Professor in the Department of Electrical and Computer Engineering at the University of California, Santa Barbara. Jason received a BS in Mechanical Engineering in 2001 from UCLA, and a PhD in Mechanical Engineering in 2007, also from UCLA, under the supervision of Jeff S. Shamma, where he was awarded the Outstanding Graduating PhD Student in Mechanical Engineering. After graduating from UCLA, he served as a junior fellow in the Social and Information Sciences Laboratory at the California Institute of Technology until 2010 when he joined the University of Colorado. Jason is a recipient of the NSF Career Award (2014), the ONR Young Investigator Award (2015), the AFOSR Young Investigator Award (2012), the American Automatic Control Council Donald P. Eckman Award (2012), and the SIAG/CST Best SICON Paper Prize (2015). Jason's research interests focus on game theoretic methods for the control of distributed multiagent systems.
\end{IEEEbiography}
\vfill 
\end{document}